\documentclass[11pt,letterpaper]{article}
\usepackage[margin=1in]{geometry}
\usepackage{authblk}

\usepackage{microtype}
\usepackage{xspace}
\usepackage{algorithm}
\usepackage{algorithmic}
\usepackage{amsthm}
\usepackage{amsmath}
\usepackage{amssymb}
\usepackage[demo]{graphicx}
\usepackage{caption}
\usepackage{subcaption}
\usepackage{url}
\usepackage{hyperref}

\usepackage{xcolor}
\usepackage[draft,inline,nomargin,index]{fixme}
\FXRegisterAuthor{sv}{asv}{\colorbox{gray!10!white}{\color{black}Suresh}}
\FXRegisterAuthor{ab}{aab}{\colorbox{blue!10!white}{\color{black}Aditya}}
\FXRegisterAuthor{sd}{asd}{\colorbox{red!10!white}{\color{black}Samira}}

\newcommand{\eps}{\varepsilon}

\newcommand{\E}{\mathbb{E}}
\newcommand{\Exp}{\mathop{\E}}
\newcommand{\Var}{\textbf{\textsf{Var}}}
\renewcommand{\Pr}{\textbf{\textsf{Pr}}}

\newcommand{\reals}{\ensuremath{\mathbb{R}}}

\newcommand{\otilde}{\widetilde{O}}
\newcommand{\akkest}{\textsc{Est}}

\newcommand{\vertexsample}{\textsc{vertex sample}\xspace}
\newcommand{\edgesample}{\textsc{edge sample}\xspace}
\newcommand{\npara}[1]{\vspace{0.05in} \noindent \textbf{ #1 }}

\newcommand{\maxcut}{\textsc{MaxCut}\xspace}

\newcommand{\suu}[1]{^{(#1)}}

\newtheorem{theorem}{Theorem}[section]
\newtheorem{corollary}{Corollary}[theorem]
\newtheorem{obs}[theorem]{Observation}
\newtheorem{lemma}[theorem]{Lemma}

\newtheorem{definition}[theorem]{Definition}

\newcommand{\cC}{\mathcal{C}}

\title{Sublinear Algorithms for MAXCUT and Correlation Clustering \thanks{This research was supported in part by National Science Foundation under grants IIS-1251049, IIS-1633724}}
\author[1]{Aditya Bhaskara}
\author[1,2]{Samira Daruki}
\author[1]{Suresh Venkatasubramanian}
\affil[1]{School of Computing, University of Utah}
\affil[2]{Expedia Research}
\date{}

\begin{document}

\maketitle

\begin{abstract}
We study sublinear algorithms for two fundamental graph problems, MAXCUT and correlation clustering. Our focus is on constructing core-sets as well as developing streaming algorithms for these problems. Constant space algorithms are known for {\em dense} graphs for these problems, while $\Omega(n)$ lower bounds exist (in the streaming setting) for sparse graphs. 

Our goal in this paper is to bridge the gap between these extremes. Our first result is to construct core-sets of size $\otilde(n^{1-\delta})$ for both the problems, on graphs with average degree $n^{\delta}$ (for any $\delta >0$). This turns out to be optimal, under the exponential time hypothesis (ETH).  Our core-set analysis is based on studying random-induced sub-problems of optimization problems. To the best of our knowledge, all the known results in our parameter range rely crucially on near-regularity assumptions. We avoid these by using a biased sampling approach, which we analyze using recent results on concentration of quadratic functions. We then show that our construction yields a 2-pass streaming $(1+\eps)$-approximation for both problems; the algorithm uses $\otilde(n^{1-\delta})$ space, for graphs of average degree $n^\delta$.

\end{abstract}

\section{Introduction}
\label{sec:intro}

Sublinear algorithms are a powerful tool for dealing with large data problems. The range of questions that can be answered accurately using sublinear (or even polylogarithmic) space or time is enormous, and the underlying techniques of sketching, streaming, sampling and core-sets have been proven to be a rich toolkit. 

When dealing with large graphs, the sublinear paradigm has yielded many powerful results. For many NP-hard problems on graphs, classic results from property testing~\cite{goldreich1998property, alon2009combinatorial} imply extremely efficient sublinear approximations. In the case of dense graphs, these results (and indeed older ones of~\cite{arora1995polynomial, frieze1996regularity}) provide constant time/space algorithms. More recently, graph sketching techniques have been used to obtain efficient approximation algorithms for cut problems on graphs~\cite{ahn2009graph, ahn2012graph} in a streaming setting. These algorithms use space that is nearly linear in $n$ (the number of vertices) and are sublinear in the number of edges as long as $|E| = \omega(n)$ (this is called the ``semi-streaming'' setting). 

By way of lower bounds, recent results have improved our understanding of the limits of sketching and streaming. In a sequence of results~\cite{kapralov2013better, kapralov2014approximating, kapralov20171}, it was shown that for problems like matching and \maxcut in a streaming setting, $\Omega(n)$ space is necessary in order to obtain any approximation better than a factor 2 in one round. (Note that a factor 2 is trivial by simply counting edges.) Furthermore, Andoni et al.~\cite{andoni2014sketching} showed that any sketch for all the cuts in a graph must have size $\Omega(n)$. 

While these lower bounds show that $O(n)$ space is the best possible for approximating problems like \maxcut in general, the constructions used in these bounds are quite specialized. In particular, the graphs involved are sparse, i.e., have $\Theta(n)$ edges.  Meanwhile, as we mentioned above, if a graph is {\em dense} ($\Omega(n^2)$ edges), random sampling is known to give $O(1)$ space and time algorithms. The question we study in this paper is if there is a middle ground: can we get truly sublinear (i.e., $o(n)$) algorithms for natural graph problems in between (easy) dense graphs and (hard) sparse graphs?

Our main contribution is to answer this in the affirmative. As long as a graph has average degree $n^{\delta}$ for some $\delta >0$, truly sub-linear space $(1+\epsilon)$ approximation algorithms are possible for problems such as \maxcut and correlation clustering. Note that we consider the max-agreement version of correlation clustering (see Section~\ref{sec:definitions})  Indeed, we show that a biased sample of vertices forms a ``core-set'' for these problems. A core-set for an optimization problem (see~\cite{agarwal2005geometric}), is a subset of the input with the property that a solution to the subset provides an approximation to the solution on the entire input.

Our arguments rely on understanding the following fundamental question: given a graph $G$, is the induced subgraph on a random subset of vertices a core-set for problems such as \maxcut? This question of sub-sampling and its effect on the value of an optimization problem is well studied. Results from property testing imply that a uniformly random sample of constant size suffices for many problems on {\em dense} graphs. \cite{frieze1996regularity, alon2003random} generalized these results to the case of arbitrary $k$-CSPs. More recently,~\cite{barak2011}, extending a result in~\cite{feige2002optimality}, studied the setting closest to ours. For graphs, their results imply that when the maximum and minimum degrees are both $\Theta(n^{\delta})$, then a random induced subgraph with $\otilde(n^{1-\delta})$ acts as a core-set for problems such as \maxcut. Moreover, they showed that for certain lifted relaxations, subsampling does {\em not} preserve the value of the objective. Finally, using more modern techniques,~\cite{rudelson2007} showed that the {\em cut norm} of a matrix (a quantity related to the \maxcut) is preserved up to a constant under random sampling, improving on~\cite{frieze1996regularity, alon2003random}. 
While powerful, we will see that these results are not general enough for our setting.  Thus we propose a new, conceptually simple technique to analyze sub-sampling, and present it in the context of \maxcut and correlation clustering. 

\subsection{Our Results}
\label{sec:our-results}

As outlined above, our main result is to show that there exist core-sets of size $\otilde(n^{1-\delta})$ for \maxcut and correlation clustering for graphs with $\Omega(n^{1+\delta})$ edges (where $0 < \delta \le 1$). This then leads to a two-pass streaming algorithm for \maxcut and correlation clustering on such graphs, that uses $\otilde(n^{1-\delta})$ space and produces a $1+\eps$ approximation.

This dependence of the core-set size on $\delta$ is optimal up to logarithmic factors, by a result of~\cite{fotakis2015sub}.  Specifically, \cite{fotakis2015sub} showed that any $(1+\eps)$ approximation algorithm for \maxcut on graphs of average degree $n^{\delta}$ must have running time $2^{\Omega(n^{1-\delta})}$, assuming the exponential time hypothesis (ETH).  Since a core-set of size $o(n^{1-\delta})$ would trivially allow such an algorithm (we can perform exhaustive search over the core-set), our construction is optimal up to a logarithmic factor, assuming ETH.

%
%
%

Our streaming algorithm for correlation clustering can be viewed as improving the semi-streaming (space $\tilde{O}(n)$) result of Ahn et al.~\cite{ahn2015correlation}, while using an additional pass over the data. Also, in the context of the lower bound of Andoni et al. \cite{andoni2014sketching}, our result for \maxcut can be interpreted as saying that while a sketch that approximately maintains \emph{all} cuts in a graph requires an $\Omega(n)$ size, one that preserves the \maxcut can be significantly smaller, when the graph has a polynomial average degree.

At a technical level, we analyze the effect of sampling on the value of the \maxcut and correlation clustering objectives. As outlined above, several techniques are known for such an analysis, but we give a new and conceptually simple framework that (a) allows one to analyze {\em non-uniform} sampling for the first time, and (b) gets over the assumptions of near-regularity (crucial for~\cite{feige2002optimality, barak2011}) and density (as in~\cite{frieze1996regularity, alon2003random}).
We expect the ideas from our analysis to be applicable to other settings as well, especially ones for which the `linearization' framework of~\cite{arora1995polynomial} is applicable.

The formal statement of results, an outline of our techniques and a comparison with earlier works are presented in Section~\ref{sec:technical-contribution}.

\subsection{Related Work}
\label{sec:prior-work}

\maxcut and correlation clustering are both extremely well-studied problems, and thus we will only mention the results most relevant to our work.

\npara{Dense graphs.} A graph is said to be dense if its average degree is $\Omega(n)$.  Starting with the work of Arora et al.~\cite{arora1995polynomial}, many NP hard optimization problems have been shown to admit a PTAS when the instances are dense. Indeed, a small random induced subgraph is known to be a core-set for problems such as \maxcut, and indeed all $k$-CSPs~\cite{goldreich1998property, alon2003random, frieze1996regularity, mathieu2008yet}. The work of~\cite{arora1995polynomial} relies on an elegant {\em linearization} procedure, while~\cite{frieze1996regularity, alon2003random} give a different (and more unified) approach based on ``cut approximations'' of a natural tensor associated with a CSP.

\npara{Polynomial density.} The focus of our work is on graphs that are {\em in between} sparse (constant average degree) and dense graphs. These are graphs whose density (i.e., average degree) is $n^{\delta}$, for some $0 < \delta < 1$.  Fotakis et al.~ \cite{fotakis2015sub} extended the approach of~\cite{arora1995polynomial} to this setting, and obtained $(1+\eps)$ approximation algorithms with run-time $\exp(\otilde(n^{1-\delta}))$.  They also showed that it was the best possible, under the exponential time hypothesis (ETH). By way of core-sets, in their celebrated work on the optimality of the Goemans-Williamson rounding, Feige and Schechtman~\cite{feige2002optimality} showed that a random sample of $\otilde(n^{1-\delta})$ is a core-set for \maxcut, if the graphs are {\em almost regular} and have an average degree $n^{\delta}$. This was extended to other CSPs by~\cite{barak2011}. These arguments seem to use near-regularity in a crucial way, and are based on restricting the number of possible `candidates' for the maximum cut.

\npara{Streaming algorithms and lower bounds.}  In the streaming setting, there are several algorithms \cite{ahn2009graph, kelner2013spectral, goel2010graph, ahn2012graph, goel2012single, kapralov2014single} that produce cut or spectral sparsifiers with $O(\frac{n}{\epsilon^2})$ edges using $\tilde{O}(\frac{n}{\epsilon^2})$ space.  Such algorithms preserves every cut within $(1+\epsilon)$-factor (and therefore also preserve the max cut).  Andoni et al. \cite{andoni2014sketching} showed that such a space complexity is essential; in fact,~\cite{andoni2014sketching} show that any sketch for all the cuts in a graph must have bit complexity $\Omega(\frac{n}{\epsilon^2})$ (not necessarily streaming ones). 
However, this does not rule out the possibility of being able to find a {\em maximum} cut in much smaller space.  

For \maxcut, Kapralov et al. \cite{kapralov2015streaming} and independently Kogan et al. \cite{bhm4} proved that any streaming algorithm that can approximate the \maxcut value to a factor better than $2$ requires $\tilde{O}(\sqrt{n})$ space, even if the edges are presented in random order.  For adversarial orders, they showed that for any  $\epsilon > 0$, a one-pass $(1+\epsilon)$-approximation to the max cut value must use $n^{1- O(\epsilon)}$ space.  Very recently, Kapralov et al. \cite{kapralov20171} went further, showing that there exists an $\epsilon^* > 0$ such that every randomized single-pass streaming algorithm that yields a $(1+\epsilon^*)$-approximation to the MAXCUT size must use $\Omega(n)$ space.

\npara{Correlation clustering.} 
Correlation clustering was formulated by Bansal et al. \cite{bansal2004correlation} and has been studied extensively.  There are two common variants of the problem -- maximizing agreement and minimizing disagreement. While these are equivalent for exact optimization (their sum is a constant), they look very different under an approximation lens. Maximizing agreement typically admits constant factor approximations, but minimizing disagreement is much harder. In this paper, we focus on the maximum-agreement variant of correlation clustering and in particular we focus on $(1+\epsilon)$-approximations. Here, Ailon and Karnin \cite{ailon2012note} presented an approximation scheme with sublinear query complexity (which also yields a semi-streaming algorithm) for dense instances of correlation clustering. Giotis and Guruswami \cite{giotis2006correlation} described a sampling based algorithm combined with a greedy strategy which guarantees a solution within $(\epsilon n^2)$ additive error.  (Their work is similar to the technique of Mathieu and Schudy \cite{mathieu2008yet}.) Most recently, Ahn et al. \cite{ahn2015correlation} gave a single-pass semi-streaming algorithm for max-agreement. For bounded weights, they provide an $(1+ \epsilon)$-approximation streaming algorithm and for graphs with arbitrary weights, they present a $0.766 (1-\epsilon)$-approximation algorithm. Both algorithms require $(n \epsilon^{-2})$ space. The key idea in their approach was to adapt multiplicative-weight-update methods for solving the natural SDPs for correlation clustering in a streaming setting using linear sketching techniques. 

\section{Definitions}
\label{sec:definitions}
\begin{definition}[\maxcut]
Let $G = (V, E, w)$ be a graph with weights $w : E \to \reals^+$. Let $(A,B)$ be a partition of $V$ and let $w(A, B)$ denote the sum of weights of edges between $A$ and $B$. Then  $\maxcut(G) = \max_{(A,B) \text{\ partition of\ }  V} w(A, B)$.
\end{definition}

For ease of exposition, we will assume that the input graph for \maxcut is unweighted. Our techniques apply as long as all the weights are $O(1)$.  Also, we denote by $\Delta$ the average degree, i.e., $\sum_{i,j} w_{ij}/|V|$. 


Moving now to correlation clustering, let $G = (V, E, c^+, c^-)$ be a graph with edge weights $c^{+}_{ij}$ and $c^{-}_{ij}$ where for every edge $ij$ we have $c^{+}_{ij}, c^{-}_{ij} \geq 0$ and only one of them is nonzero. For every edge $ij \in E$, we define $\eta_{ij} = c^{+}_{ij} - c^{-}_{ij}$ and for each vertex, $d_i = \sum_{i \in \Gamma(j)} |\eta_{ij}|$. We will also assume that all the weights are bounded by an absolute constant in magnitude (for simplicity, we assume it is $1$).  We define the ``average degree'' $\Delta$ (used in the statements that follow) of a correlation clustering instance to be $(\sum_i d_i  )/n$.

\begin{definition}[MAX-AGREE correlation clustering]
Given $G = (V, E, c^+, c^-)$ as above, consider a partition of $V$ into \emph{clusters} $C_1, C_2, \dots, C_k$, and let $\chi_{ij}$ be an indicator that is $1$ if $i$ an $j$ are in the same cluster and $0$ otherwise. The MAX-AGREE score of this clustering is given by
$\sum_{ij} c_{ij}^+ \chi_{ij} + \sum_{ij}  c_{ij}^- (1-\chi_{ij})$.
The goal is to find a partition \emph{maximizing} this score. The maximum value of the score over all partitions of $V$ will be denoted by $CC(G)$.
\end{definition}

Note that the objective value can be simplified to $\sum_{ij} c_{ij}^- + \eta_{ij} \chi_{ij} = C^- + \sum_{ij} \eta_{ij} \chi_{ij}$, where $C^-$ denotes the sum $\sum_{ij} c_{ij}^-$. 

We will also frequently use concentration bounds, which we state next.

\section{Preliminaries}
\label{sec:preliminaries}
We will frequently appeal to Bernstein's inequality for concentration of linear forms of random variables. For completeness, we state it here.
\begin{theorem}[Bernstein's inequality\cite{dubhashi2009concentration}] \label{bern}
Let the random variables $X_1, \cdots, X_n$ be independent with $|X_i - E[X_i]| \leq b$ for each $i \in [n]$. Let $X= \sum_i X_i$ and let $\sigma^2 = \sum_i \sigma_i^2$ be the variance of $X$. Then, for any $t>0$, 
\begin{align*}
\Pr[|X - \E[X] | > t] \leq \exp (- \frac{t^2}{2 \sigma^2 (1+ bt/3\sigma^2)})
\end{align*}
\end{theorem}

A slightly more non-standard concentration inequality we use is from Boucheron, Massart and Lugosi~\cite{boucheron2003concentration}. It can be viewed as an exponential version of the classic Efron-Stein lemma.

\begin{theorem}[\cite{boucheron2003concentration}]\label{concentration-main}
Assume that $Y_1, \cdots, Y_n$ are random variables, and $Y_1^n$ is the vector of these $n$ random variables. Let $Z = f(Y_1, \cdots, Y_n)$, where $f: \chi_n \rightarrow R$ is a measurable function. Define $Z^{(i )} = f(Y_1, \cdots, Y_{i-1}, Y_i', Y_{i+1}, \cdots, Y_n)$, where $Y_1, \cdots, Y_n'$ denote the independent copies of $Y_1, \cdots, Y_n$. Then, for all $\theta > 0$ and $\lambda \in (0, \frac{1}{\theta})$, 
\begin{align*}
\log \E[e^{\lambda (Z - \E[Z])}] \leq \frac{\lambda \theta}{1- \lambda \theta} \log \E[e^{\frac{\lambda L_{+}}{\theta}}],
\end{align*}
where $L_{+}$ is the random variable defined as
\begin{align*}
L_{+} = \E[\sum_{i = 1}^n {(Z - Z^{(i)})}^2 \mathbf{ 1}_{Z > Z^{(i)}} | Y_1^n].
\end{align*}
\end{theorem}


\section{Technical overview}
\label{sec:technical-contribution}

We now present an outline of our main ideas. Suppose we have a graph $G = (V, E)$.  First, we define a procedure \vertexsample.  This takes as input probabilities $p_i$ for every vertex, and produces a random weighted induced subgraph. 

\npara{Procedure  \vertexsample$(\{p_i\}_{i \in V})$.}
Sample a set $S'$ of vertices by selecting each vertex $v_i$ with probability $p_i$ independently. Define $H$ to be the induced subgraph of $G$ on the vertex set $S'$. For $i, j \in S'$, define $w_{ij} = 
\frac{1}{p_i p_j \Delta^2}$.\footnote{In correlation clustering, we have edge weights to start with, so the weight in $H$ will be $w_{ij} \cdot c_{ij}^+$ (or $c_{ij}^-$).}

Intuitively, the edge weights are chosen so that the total number of edges remains the same, in expectation.  Next, we define the notion of an importance score for vertices. Let $d_i$ denote the degree of vertex $i$. 
\begin{definition}\label{score}
The \emph{importance score} $h_i$ of a vertex $i$ is defined as $h_i = \min \{1, \frac{\max \{d_i , \epsilon \Delta \}}{\Delta^2 \alpha_{\epsilon}} \}$,  where $\alpha_{\epsilon}$ is an appropriately chosen parameter (for $\maxcut$, we set it to $\frac{\epsilon^4}{C \log n}$, and for correlation clustering, we set it to $\frac{\epsilon^8}{C \log n}$, where $C$ is an absolute constant).
\end{definition}




The main result is now the following:
\begin{theorem}[Core-set] \label{sampling}
Let $G = (V, E)$ have an average degree $\Delta$. Suppose we apply \vertexsample with probabilities $p_i \in [h_i, 2h_i]$ to obtain a weighted graph $H$.  Then $H$ has $\otilde(\frac{n}{\Delta})$ vertices and the quantities $\maxcut(H)$ and $CC(H)$ are within a $(1+ \epsilon)$ factor of the corresponding quantities $\maxcut(G)$ and $CC(G)$, w.p. at least $1- \frac{1}{n^2}$.
\end{theorem}

While the number of vertices output by the vertex sample procedure is small, we would like a core-set of small ``total size''. This is ensured by the following.

\npara{Procedure \edgesample $(H)$.} Given a weighted graph $H$ with total edge weight $W$, sample each edge $e \in E(H)$ independently with probability $p_{e} :=  \min(1, \frac{8 |S'| w_{e}}{\eps^2 W})$, to obtain a graph $H'$.  Now, assign a weight $w_e/p_e$ to the edge $e$ in $H'$.

The procedure samples roughly $|S'|/\eps^2$ edges, with probability proportional to the edge weights.  The graph is then re-weighted in order to preserve the total edge weight in expectation, yielding:

\begin{theorem}[Sparse core-set] \label{coreset}
Let $G$ be a graph $n$ vertices and average degree $\Delta = n^{\delta}$.  Let $H'$ be the graph obtained by first applying \vertexsample and then applying \edgesample. Then $H'$ is a $\epsilon$-core-set for $\maxcut$ and $CC$, having size $\widetilde{O}(\frac{n}{\Delta}) = \widetilde{O}(n^{1 - \delta})$.
\end{theorem}

We then show how to implement the above procedures in a streaming setting. This gives: 
\begin{theorem}[Streaming algorithm]\label{streaming}
Let $G$ be a graph on $n$ vertices and average degree $\Delta = n^{\delta}$, whose edges arrive in a streaming fashion in adversarial order. There is a two-pass streaming algorithm with space complexity $\otilde(\frac{n}{\Delta}) = \otilde(n^{1 - \delta})$ for computing a $(1+ \epsilon)$-approximation to $\maxcut(G)$ and $CC(G)$.
\end{theorem}

Of these, Theorem~\ref{sampling} is technically the most challenging. Theorem~\ref{coreset} follows via standard edge sampling methods akin to those in \cite{ahn2009graph} (which show that w.h.p., every cut size is preserved). It is presented in Section~\ref{sec:edge-sample}, for completeness. The streaming algorithm, and a proof of Theorem~\ref{streaming}, are presented in Section~\ref{sec:2-pass-streaming}. In the following section, we give an outline of the proof of Theorem~\ref{sampling}.

\subsection{Proof of the sampling result (theorem~\ref{sampling}): an outline}
\label{sec:paper-outline}

In this outline we will restrict ourselves to the case of \maxcut as it illustrates our main ideas.  Let $G$ be a graph as in the statement of the theorem, and let $H$ be the output of the procedure \vertexsample.

Showing that $\maxcut(H)$ is at least $\maxcut(G)$ up to an $\eps n \Delta$ additive term is  easy. We simply look at the projection of the maximum cut in $G$ to $H$ (see, for instance,~\cite{feige2002optimality}). Thus, the challenge is to show that a sub-sample cannot have a significantly larger cut, w.h.p. The natural approach of showing that every cut in $G$ is preserved does not work as $2^n$ cuts is too many for the purposes of a union bound.

There are two known ways to overcome this.  
The first approach is the one used in~\cite{goldreich1998property, feige2002optimality} and~\cite{barak2011}.  These works essentially show that in a graph of average degree $\Delta$, we need to consider only roughly $2^{n/\Delta}$ cuts for the union bound.  If all the degrees are roughly $\Delta$, then one can show that all these cuts are indeed preserved, w.h.p.  There are two limitations of this argument. First, for non-regular graphs, the variance (roughly $\sum_i p d_i^2$, where $p$ is the sampling probability) can be large, and we cannot take a union bound over $\exp(n/\Delta)$ cuts. Second, the argument is combinatorial, and it seems difficult to generalize this to analyze non-uniform sampling.

The second approach is via {\em cut decompositions}, developed in~\cite{frieze1996regularity, alon2003random}.  Here, the adjacency matrix $A$ is decomposed into $\text{poly}(1/\eps)$ rank-1 matrices, plus a matrix that has a small {\em cut norm}.  It turns out that solving many quadratic optimization problems (including \maxcut) on $A$ is equivalent (up to an additive $\eps n\Delta$) to solving them over the sum of rank-1 terms (call this $A'$). Now, the adjacency matrix of $H$ is an induced square sub-matrix of $A$, and since we care only about $A'$ (which has a simple structure),~\cite{alon2003random} could show that $\maxcut(H) \le \maxcut(G) + \eps n^2$, w.h.p. To the best of our knowledge, such a result is not known in the ``polynomial density'' regime (though the cut decomposition still exists). 

\npara{Our technique.}  We consider a new approach.  While inspired by ideas from the works above, it also allows us to reason about non-uniform sampling in the polynomial density regime.  Our starting point is the result of Arora et al.~\cite{arora1995polynomial}, which gives a method to estimate the \maxcut using a collection of linear programs (which are, in turn, derived using a sample of size $n/\Delta$).  Now, by a double sampling trick (which is also used in the approaches above), it turns out that showing a sharp concentration bound for the value of an {\em induced sub-program} of an LP as above, implies Theorem~\ref{sampling}.  As it goes via a linear programming and not a combinatorial argument, analyzing non-uniform sampling turns out to be quite direct.  Let us now elaborate on this high level plan.

\npara{Induced sub-programs.}  First, we point out that an analysis of induced sub-programs is also an essential idea in the work of~\cite{alon2003random}. The main difference is that in their setting, only the variables are sub-sampled (and the number of constraints remains the same). In our LPs, the constraints correspond to the vertices, and thus there are fewer constraints in the sampled LP.  This makes it harder to control the value of the objective.  At a technical level, while a duality-based argument using Chernoff bounds for linear forms suffices in the setting of~\cite{alon2003random}, we need the more recent machinery on concentration of quadratic functions.

We start by discussing the estimation technique of~\cite{arora1995polynomial}.


%
\npara{Estimation with Linear Programs.}
The rough idea is to start with the natural quadratic program for \maxcut: $\max \sum_{(i,j) \in E} x_i (1-x_j)$, subject to $x_i \in \{0,1\}$.\footnote{This is a valid formulation, because for every $x_i \ne x_j$ that is an edge contributes $1$ to the objective, and $x_i = x_j$ contribute $0$.}  This is then ``linearized'' using 
a seed set of vertices sampled from $G$.  We refer to  Section~\ref{sec:gener-estim-meth} for details.  For now, $\textsc{Est}(G)$ is a procedure that takes a graph $G$ and a set of probabilities $\{ \gamma_i \}_{i \in V(G)}$, samples a seed set using $\gamma$, and produces an estimate of \maxcut$(G)$. 

Now, suppose we have a graph $G$ and a sample $H$.  We can imagine running $\textsc{Est}(G)$ and $\textsc{Est}(H)$ to obtain good estimates of the respective \maxcut values.  But now suppose that in both cases, we could use precisely the same seed set.  Then, it turns out that the LPs used in $\textsc{Est}(H)$ would be `induced' sub-programs (in a sense we will detail in Section~\ref{sec:dual}) of those used in $\textsc{Est}(H)$, and thus proving Theorem~\ref{sampling} reduces to showing a sufficiently strong concentration inequality for sub-programs.

The key step above was the ability to use same seed set in the \textsc{Est} procedures.  This can be formalized as follows. 

\npara{Double sampling.} Consider the following two strategies for sampling a \emph{pair} of subsets $(S, S')$ of a universe $[n]$ (here, $q_v \le p_v$ for all $v$):
\begin{itemize}
\item Strategy A:  choose $S' \subseteq [n]$, by including every $v$ w.p. $p_v$, independently; then for $v \in S'$, include them in $S$ w.p. $q_v/p_v$, independently.
\item Strategy B:  pick $S \subseteq [n]$, by including every $v$ w.p. $q_v$; then iterate over $[n]$ once again, placing $v \in S'$ with a probability equal to $1$ if $v\in S$, and $p_v^\ast$ if $v \not\in S$. 
\end{itemize}

\begin{lemma} \label{strategy1-2}
Suppose $p^\ast_v = p_v (1 - \frac{q_v}{p_v})/(1 - p_v)$. Then the distribution on pairs $(S, S')$ obtained by strategies A and B are identical.
\end{lemma}
The proof is by a direct calculation, which we state it here.
\begin{proof}
Let us examine strategy A. It is clear that the distribution over $S$ is precisely the same as the one obtained by strategy B, since in both the cases, every $v$ is included in $S$ independently of the other $v$, with probability precisely $q_v$.  Now, to understand the joint distribution $(S, S')$, we need to consider the conditional distribution of $S'$ given $S$. Firstly, note that in both strategies, $S \subseteq S'$, i.e., $\Pr[ v \in S' | v \in S] =1$. Next, we can write $\Pr_{\text{strategy A}} [ v \in S' | v \not\in S]$ as
\[ \frac{\Pr_{\text{strategy A}} [ v \in S' \wedge v \not\in S ]}{\Pr_{\text{strategy A}} [ v \not\in S]} = \frac{p_v (1-\frac{q_v}{p_v})}{1- p_v}. \]

Noting that $\Pr_{\text{strategy} B} [ v \in S' | v \not\in S] = p^\ast_v$ (by definition) concludes the proof.
\end{proof}

\npara{Proof of Theorem~\ref{sampling}. }   To show the theorem, we use $p_v$ as in the statement of the theorem, and set $q$ to be the uniform distribution $q_v = \frac{16 \log n}{\eps^2 \Delta}$.  The proof now proceeds as follows.  Let $S'$ be a set sampled using the probabilities $p_v$. These form the vertex set of $H$.  Now, the procedure $\akkest$ on $H$ (with sampling probabilities $q_v/p_v$) samples the set $S$ (as in strategy A).  By the guarantee of the estimation procedure (Corollary~\ref{condition-H}), we have $\maxcut(H) \approx \akkest(H)$, w.h.p.   Next, consider the procedure \akkest{} on $G$ with sampling probabilities $q_v$.  Again, by the guarantee of the estimation procedure (Corollary~\ref{condition-G}), we have $\maxcut(G) \approx \akkest(G)$, w.h.p.  

Now, we wish to show that $\akkest(G) \approx \akkest(H)$.  By the equivalence of the sampling strategies, we can now take the strategy B view above.  This allows us to assume that the $\akkest$ procedures use the same $S$, and that we pick $S'$ {\em after} picking $S$. This reduces our goal to one of analyzing the value of a random induced sub-program of an LP, as mentioned earlier.  The details of this step are technically the most involved, and are presented in Section~\ref{sec:dual}.  This completes the proof of the theorem. (Note that the statement also includes a bound on the number of vertices of $H$.  This follows immediately from the choice of $p_v$.) \qed

\section{Estimation via linear programming}
\label{sec:gener-estim-meth}

We now present the estimation procedure $\akkest$ used in our proof.  It is an extension of~\cite{arora1995polynomial} to the case of weighted graphs and non-uniform sampling probabilities. 

Let $H = (V, E, w)$ be a weighted, undirected graph with edge weights $w_{ij}$, and let $\gamma: V \rightarrow [0,1]$ denote sampling probabilities.  The starting point is the quadratic program for \maxcut: $\max ~\sum_{ij \in E} w_{ij} x_i (1-x_j)$, subject to $x_i \in \{0,1\}$.  The objective can be re-written as $\sum_{i \in V} x_i (d_i - \sum_{j \in \Gamma(i)} w_{ij} x_j)$, where $d_i$ is the weighted degree, $\sum_{j \in \Gamma(i)} w_{ij}$.  The key idea now is to ``guess'' the value of $\rho_i := \sum_{j \in \Gamma(i)} w_{ij} x_j$, by using a seed set of vertices.  Given a guess, the idea is to solve the following linear program, which we denote by $LP_\rho(V)$.
\begin{align*}\label{eq:akk}
	\text{maximize} &\quad \sum_i x_i (d_i - \rho_i) - s_i -
	t_i \\
	\text{subject to} &\quad \rho_i - t_i \le \sum_{j \in \Gamma(i)} w_{ij} x_j \le
	\rho_i + s_i\\ 
	\qquad 0 \le x_i \le 1,& \qquad s_i, t_i \ge 0.
\end{align*}
	
The variables are $x_i, s_i, t_i$.  Note that if we fix the $x_i$, the optimal $s_i, t_i$ will satisfy $s_i + t_i = |\rho_i - \sum_{j \in \Gamma(i)} w_{ij} x_j|$.   Also, note that if we have a perfect guess for $\rho_i$'s (coming from the \maxcut), the objective can be made $\ge \maxcut(H)$.

\npara{Estimation procedure.}  The procedure $\akkest$ is the following:  first sample a set $S \subseteq V$ where each $i \in V$ is included w.p. $\gamma_i$ independently.  For every partition $(A, S\setminus A)$ of $S$, set $\rho_i = \sum_{j \in \Gamma(i) \cap A} \frac{w_{ij}}{\gamma_j}$, and solve $LP_{\rho} (V)$  (in what follows, we denote this LP by $LP_{A, S\setminus A}^{\gamma}(V)$, as this makes the partition and the sampling probabilities clear). Return the maximum of the objective values. 


%
%

Our result here is a sufficient condition for having $\akkest(H) \approx \maxcut(H)$.

\begin{theorem} \label{estimation}
Let $H$ be a weighted graph on $n$ vertices, with edge weights $w_{ij}$ that add up to $W$.  Suppose the sampling probabilities $\gamma_i$ satisfy the condition
\begin{equation}
\label{eq:w-condition}w_{ij} \le \frac{W \eps^2}{8 \log n} \frac{\gamma_i \gamma_j}{\sum_u \gamma_u} \quad \text{for all $i,j$.}
\end{equation}
Then, we have $\akkest(H, \gamma) \in \maxcut (H) \pm \eps W$, with probability at least $1 - 1/n^2$ (where the probability is over the random choice of $S$).
\end{theorem}

The proof of the Theorem consists of claims showing the upper and lower bound separately. 

\npara{Claim 1.} The estimate is not too small. I.e., w.h.p. over the choice of $S$, there exists a cut $(A, S\setminus A)$ of $S$ such that $LP_{A, S\setminus A}(H) \ge \maxcut(H) - \eps W$.

\npara{Claim 2.} The estimate is not much larger than an optimal cut. Formally, for any feasible solution to the LP (and indeed any values $\rho_i$), there is a cut in $H$ of value at least the LP objective.

\begin{proof}[Proof of Claim 1.]
Let $(A_H, V\setminus A_H)$ be the max cut in the full graph $H$. Now consider a sample $S$, and let $(A, S \setminus A)$ be its projection onto $S$. For any vertex $i$, recall that $\rho_i = \sum_{j \in \Gamma(i) \cap A} \frac{w_{ij}}{\gamma_j} = \sum_{j \in \Gamma(i) \cap A_H} Y_j \frac{w_{ij}}{\gamma_j}$, where $Y_j$ is the indicator for $j \in S$. Thus
\[ \E [\rho_i] = \sum_{j \in \Gamma(i) \cap A_H} \gamma_j \frac{w_{ij}}{\gamma_j} = \sum_{j \in \Gamma(i) \cap A_H} w_{ij}.\]

We will use Bernstein's inequality to bound the deviation in $\rho_i$ from its mean. To this end, note that the variance can be bounded as
\[ \Var[\rho_i] = \sum_{j \in \Gamma(i) \cap A_H} \gamma_j (1-\gamma_j) \frac{w_{ij}^2}{\gamma_j^2} \le \sum_{j \in \Gamma(i)} \frac{ (1-\gamma_j) w_{ij}^2}{\gamma_j}. \]

In what follows, let us write $d_i = \sum_{j \in \Gamma(i)} w_{ij}$ and $f_i = \frac{W \gamma_i}{\sum_u \gamma_u}$. Then, for every $j$, our assumption on the $w_{ij}$ implies that $\frac{w_{ij}}{\gamma_j} \le \frac{\eps^2}{8\log n} f_i$.  Thus, summing over $j$, we can bound the variance by $\frac{\eps^2 d_i f_i}{8 \log n}$. Now, using Bernstein's inequality (Theorem \ref{bern}), 
\begin{equation}
\Pr[|\rho_i - \E[\rho_i]| > t]  \leq \exp \left(- \frac{t^2}{\frac{\eps^2 d_i f_i}{4 \log n} + \frac{2t}{3} \frac{\eps^2 f_i}{8 \log n}} \right). 
\end{equation}
Setting $t = \epsilon (d_i + f_i)$, and simplifying, we have that the probability above is $< \exp(-4 \log n) = \frac{1}{n^4}$.
Thus, we can take a union bound over all $i \in V$, and conclude that w.p. $\ge 1- \frac{1}{n^3}$,
\begin{equation}\label{eq:stbound}
\left| \rho_i - \sum_{j \in \Gamma(i) \cap A_H} w_{ij} \right| \le \eps(d_i + f_i) \quad \text{for all $i \in V$.}
\end{equation}
For any $S$ that satisfies the above, consider the solution $x$ that sets $x_i = 1$ for $i \in A_H$ and $0$ otherwise. We can choose $s_i + t_i = |\rho_i - \sum_{j \in \Gamma(i)} w_{ij} x_j| \le \eps(d_i + f_i)$, by the above reasoning (eq.~\eqref{eq:stbound}). Thus the LP objective can be lower bounded as
\[ \sum_i x_i (d_i - \rho_i) - \eps (d_i + f_i) \ge \sum_i x_i (d_i - \sum_{j \in \Gamma(i)} w_{ij} x_j) - 2\eps(d_i + f_i).\]
This is precisely $\maxcut(G) - 2 \eps \sum_i (d_i + f_i) \ge \maxcut(G) - 4 \eps W$. This completes the proof of the claim.
\end{proof}

\begin{proof}[Proof of Claim 2.]
Suppose we have a feasible solution $x$ to the LP, of objective value $\sum_i x_i (d_i - \rho_i) - \sum_{i} |\rho_i - \sum_{j \in \Gamma(i)} w_{ij} x_j|$, and we wish to move to a cut of at least this value. To this end, define the quadratic form
\[ Q(x) := \sum_i x_i \big(d_i - \sum_{j \in \Gamma(i)} w_{ij} x_j \big). \]

The first observation is that for any $x \in [0,1]^n$, and any real numbers $\rho_i$, we have
\[ Q(x) \ge \sum_i x_i (d_i - \rho_i) - \sum_{i} |\rho_i - \sum_{j \in \Gamma(i)} w_{ij} x_j|.\]
This is true simply because $Q(x) = \sum_i x_i \big( d_i - \rho_i \big) + x_i \big( \rho_i - \sum_{j \in \Gamma(i)} w_{ij} x_j \big)$, and the fact that the second term is at least $- |\rho_i - \sum_{j \in \Gamma(i)} w_{ij} x_j|$, as $x_i \in [0,1]$. 

Next, note that the maximum of the form $Q(x)$ over $[0,1]^n$ has to occur at a boundary point, since for any fixing of variables other than a given $x_i$, the form reduces to a linear function of $x_i$, which attains maximum at one of the boundaries. Using this observation repeatedly lets us conclude that there is a $y \in \{0,1\}^n$ such that $Q(y) \ge Q(x)$.  Since any such $y$ corresponds to a cut, and $Q(y)$ corresponds to the cut value, the claim follows.\footnote{We note that the proof in~\cite{arora1995polynomial} used randomized rounding to conclude this claim, but this argument is simpler; also, later papers such as~\cite{fotakis2015sub} used such arguments for derandomization.}
\end{proof}

Finally, to show Theorem~\ref{sampling} (as outlined in Section~\ref{sec:paper-outline}), we need to apply Theorem~\ref{estimation} with specific values for $\gamma$ and $w_{ij}$. 
Here we state two related corollaries to Theorem~\ref{estimation} that imply good
estimates for the \maxcut.
\begin{corollary}\label{condition-G}
	Let $H$ in the framework be the original graph $G$, and let $\gamma_i = \frac{16
		\log n}{\eps^2 \Delta}$ for all $i$. Then the condition $w_{ij} \le
	\frac{\eps^2}{8 \log n} \cdot \frac{W \gamma_i \gamma_j}{\sum_u \gamma_u}$ holds
	for all $i, j$, and therefore $\akkest(G, \gamma) \in \maxcut(G) \pm \epsilon W$, w.p. $\ge 1-n^{-2}$.
\end{corollary}
The proof is immediate (with a slack of $2$), as $w_{ij}=1$, $W = n \Delta$, and all $\gamma_u$ are equal. 
\begin{corollary}\label{condition-H}
	Let $H$ be the weighted sampled graph obtained from \vertexsample, and let
	$\gamma_i = \frac{16 \log n}{\eps^2 \Delta} \frac{1}{p_i}$. Then the
	condition~\eqref{eq:w-condition} holds w.p. $\ge 1-n^{-3}$, and therefore $\akkest(H, \gamma) \in \maxcut(H) \pm \epsilon W$ w.p. $\ge 1-n^{-2}$.
\end{corollary}
\begin{proof}
	In this case, we have $w_{ij} = \frac{1}{p_i p_j \Delta^2}$. Thus, simplifying the condition, we need to show that
	\[ \frac{1}{p_i p_j \Delta^2} \le \frac{2 W}{p_i p_j \Delta} \frac{1}{\sum_{u \in H} \frac{1}{p_u}}. \]
	
	Now, for $H$ sampled via probabilities $p_i$, we have (in expectation) $W = \frac{n}{\Delta}$, and $\sum_{u \in H} \frac{1}{p_u} = n$. A straightforward application of Bernstein's inequality yields that $W \ge \frac{n}{2\Delta}$ and $\sum_{u \in H} \frac{1}{p_u} \le 2n$, w.p. at least $1-n^{-3}$.  This completes the proof.
\end{proof}


\newcommand{\rhotil}{\widetilde{\rho}}
\newcommand{\dtil}{\widetilde{d}}
\newcommand{\dual}{\textsf{Dual}}

\section{Random induced linear programs}
\label{sec:dual}

We will now show that the $\akkest$ on $H$ has approximately the same value as the estimate on $G$ (with appropriate $\gamma$ values). First, note that $\akkest(G)$ is $\max_{A \subseteq S} LP_{A, S \setminus A}^{\gamma}(G)$, where $\gamma_i = q_i$. To write the LP, we need the constants $\rho_i$, defined by the partition $(A, S\setminus A)$ as $\rho_i := \sum_{j \in \Gamma(i) \cap A} \frac{1}{q_j}$. For the graph $H$, the estimation procedure uses an identical program, but the sampling probabilities are now $\alpha_i := q_i/p_i$, and the estimates $\rho$, which we now denote by $\rhotil_i$, are defined by $\rhotil_i := \sum_{j \in \Gamma(i) \cap A} \frac{p_j w_{ij}}{q_j}$. 
Also, by the way we defined $w_{ij}$, $\rhotil_i = \frac{\rho_i}{p_i \Delta^2}$.  The degrees are now $\dtil_i := \sum_{j \in \Gamma(i) \cap S'} w_{ij} = \sum_{j \in \Gamma(i) \cap S'} \frac{1}{p_i p_j \Delta^2}$. The two LPs are shown in Figure~\ref{fig:primal-H}.

\begin{figure}[htbp]
	  \vspace{-.2in}
  \begin{subfigure}[b]{0.4\textwidth}
\[    \begin{aligned}
    \text{max} \quad \sum_{i \in G} [x_i (d_i - \rho_i) -  (s_i &+ t_i)] \notag\\
    \text{s.t.} \quad \sum_{j \in \Gamma(i)} x_{j} \leq \rho_i + s_i,  \quad &\forall i \in [n] \notag\\
    - \sum_{j \in \Gamma(i)} x_{j} \leq - \rho_i + t_i, \quad &\forall i \in [n] \notag\\
    0 \leq x_i \leq 1 \quad &\forall i \in [n] \notag
  \end{aligned} \]
\caption{The LP on the full graph}
\end{subfigure} \qquad
\begin{subfigure}[b]{0.4\textwidth}
  \[     \begin{aligned}
\text{max} \quad \sum_{i \in S'} [x_i (\dtil_i - \rhotil_i) - (\tilde{s}_i +& \tilde{t}_i)]\\
\text{s.t.} \quad \sum_{j \in \Gamma(i) \cap S'} w_{ij}x_{j} \leq \rhotil_i + \tilde{s}_i,  \quad &\forall i \in S'\\
   - \sum_{j \in \Gamma(i) \cap S'} w_{ij} x_{j} \leq - \rhotil_i + \tilde{t}_i, \quad &\forall i \in S'\\  0 \leq x_i \leq 1, \quad \tilde{s}_i, \tilde{t}_i \geq 0  \quad &\forall i \in S'
 \end{aligned}  \]
\caption{The sampled LP}
\end{subfigure}
\caption{The two LPs.\label{fig:primal-H}}
\end{figure}

Our aim in this section is to show the following:
\begin{theorem}\label{thm:main-lp-sample}
Let $G$ be an input graph, and let $(S, S')$ be sampled as described in Section~\ref{sec:paper-outline}. Then, with probability $\ge 1-\frac{1}{n^2}$, we have 
\[ \max_{A \subseteq S} LP_{A, S \setminus A}^\gamma (G) \ge \Delta^2 \cdot \max_{A \subseteq S} LP_{A, S \setminus A}^\alpha (H) - \eps n\Delta. \]
\end{theorem}

\noindent {\em Proof outline.} To prove the theorem, the idea is to take the ``strategy B'' viewpoint of sampling $(S, S')$, i.e., fix $S$, and sample $S'$ using the probabilities $p^\ast$.  Then, we only need to understand the behavior of an ``induced sub-program'' sampled with the probabilities $p^\ast$. This is done by considering the duals of the LPs, and constructing a feasible solution to the induced dual whose cost is not much larger than the dual of the full program, w.h.p. This implies the result, by linear programming duality.

Let us thus start by understanding the dual of $LP_{A, S\setminus A}^\gamma (G)$ given $A$, shown in Figure~\ref{eq:dualG}. We note that for any given $z$, the optimal choice of $u_i$ is $\max\{ 0, d_i - \rho_i - \sum_{j \in \Gamma(i)} z_j \}$; thus we can think of the dual solution as being the vector $z$.  The optimal $u_i$ may thus be bounded by $2d_i$, a fact that we will use later. Next, we write down the dual of the induced program, $LP_{A, S\setminus A}^\alpha (H)$, as shown in Figure~\ref{eq:dualH}. 


\begin{figure}[htbp]
  \vspace{-.2in}
  \centering
  \begin{subfigure}{0.4\textwidth}
\begin{align*}
\text{minimize} \quad \sum_{i \in G} u_i + \rho_i z_i  \quad \text{ s.t.~~~~}\\
\quad u_i + \sum_{j \in \Gamma(i)} z_j \geq d_i - \rho_i  \quad \forall i \in V \notag\\
u_i \geq 0, \quad -1 \le z_i \leq 1 \quad \forall i \in V \notag
\end{align*}
\caption{The dual of $LP_{A, S\setminus A}^\gamma (G)$ \label{eq:dualG}}
\end{subfigure}\qquad
\begin{subfigure}{0.4\textwidth}
\begin{align*}
\text{minimize} \quad \sum_{i \in S'} [\tilde{u}_i + \rhotil_i \tilde{z}_i] \quad \text{ s.t.~~~~}\\
\quad \tilde{u}_i + \sum_{j \in \Gamma(i) \cap S'} w_{ij} \tilde{z}_j \geq \dtil_i - \rhotil_i  \quad \forall i \in S' \notag \\
\tilde{u}_i \ge 0, \quad  -1 \le \tilde{z}_i \le 1 \quad \forall i \in S'.\notag
\end{align*}
\caption{The dual of the induced program $LP_{A, S\setminus A}^\alpha (H)$. \label{eq:dualH}}
\end{subfigure}
  \caption{The dual LPs}
  \label{fig:duals}
\end{figure}

\newcommand{\gs}{^{\ast} }
Following the outline above, we will construct a feasible solution to LP~\eqref{eq:dualH}, whose cost is close to the optimal dual solution to LP~\eqref{eq:dualG}. The construction we consider is very simple: if $z$ is the optimal dual solution to~\eqref{eq:dualG}, we set $\tilde{z}_i = z_i$ for $i \in S'$ as the candidate solution to~\eqref{eq:dualH}. This is clearly feasible, and thus we only need to compare the solution costs. The dual objective values are as follows
\begin{align}
\dual_G &= \sum_{i \in V} \rho_i z_i + \max\{ 0, ~ d_i - \rho_i - \sum_{j \in \Gamma(i)} z_j \}  \label{eq:dual-g} \\
\dual_H &\le \sum_{i \in S'} \rhotil_i z_i + \max\{ 0, ~ \dtil_i - \rhotil_i - \sum_{j \in \Gamma(i) \cap S'} w_{ij} z_j \}  \label{eq:dual-h}
\end{align}
Note that there is a $\le$ in~\eqref{eq:dual-h}, as $\tilde{z}_i = z_i$ is simply one feasible solution to the dual (which is a minimization program). Next, our goal is to prove that w.p. at least $1-\frac{1}{n^2}$,
\[ \max_{A \subseteq S} \dual_H \le \frac{1}{\Delta^2} \cdot \max_{A \subseteq S} \dual_G + \frac{\eps n}{\Delta}. \]

Note that here, the probability is over the choice of $S'$ {\em given} $S$ (as we are taking view-B of the sampling). The first step in proving the above is to move to a slight variant of the quantity $\dual_H$, which is motivated by the fact that $\Pr[Y_i = 1]$ is not quite $p_i$, but $p_i\gs$ (as we have conditioned on $S$). Let us define $\rhotil_i\gs := \frac{\rho_i}{p_i\gs \Delta^2}$ (recall that $\rhotil_i$ is $\frac{\rho_i}{p_i \Delta^2}$), and $w_{ij}\gs := \frac{1}{p_i\gs p_j\gs \Delta^2}$. So also, let $d_i\gs := \sum_{j \in \Gamma(i)} Y_j w_{ij}\gs$.  Then, define
\begin{equation}
\label{eq:dual-h-gs}
\dual_H\gs := \sum_{i \in S'} \rhotil_i\gs z_i + \max\{ 0, ~ \dtil_i\gs - \rhotil_i\gs - \sum_{j \in \Gamma(i) \cap S'} w_{ij}\gs z_j \} .
\end{equation}

A straightforward lemma, which we use here, is the following. Here we bound the difference between the ``corrected'' dual we used to analyze, and the value we need for the main theorem. Specifically, we bound  $| \dual_H - \dual_H\gs | \le \frac{\eps n}{2\Delta}$. 
\begin{lemma}\label{lem:dual-error}
Let $(S, S')$ be sampled as in Section~\ref{sec:paper-outline}. Then w.p. at least $1-\frac{1}{n^4}$, we have that for all $z \in [-1, 1]^n$ and for all partitions $(A, S\setminus A)$ of $S$,\footnote{Note that the partition defines the $\rho_i$.} 
$ | \dual_H - \dual_H\gs | \le \frac{\eps n}{2\Delta}$.
\end{lemma}
\begin{proof}
To prove the lemma, it suffices to prove that w.p. $\ge 1-\frac{1}{n^4}$,
\begin{equation}\label{eq:to-prove-error}
\sum_i Y_i | \rhotil_i - \rhotil_i\gs | + \sum_i Y_i \sum_{j \in \Gamma(i)} Y_j | w_{ij} - w_{ij}\gs | \le \frac{\eps n}{2 \Delta }.
\end{equation}
This is simply by using the fact that $z_i$ are always in $[-1,1]$. 
Before showing this, we introduce some notation and make some simple observations. First, denote by $Y$ the indicator vector for $S'$ and by $X$ the indicator for $S$. 
\begin{obs}
With probability $\ge 1- \frac{1}{n^4}$ over the choice of $(S, S')$, we have:
\begin{enumerate}
\item For all $i \in V$, $\sum_{j \in \Gamma(i)} \frac{X_j}{q_j} \le 2(d_i + \eps \Delta)$. 
\item For all $i \in V$, $\sum_{j \in \Gamma(i),  j \not\in S} \frac{Y_j}{p_j\gs} \le 2(d_i +\eps \Delta)$.
\item $\sum_i \frac{X_i (d_i +\eps \Delta)}{p_i} \le 2 \eps^2 n\Delta$.
\item $\sum_{i \not\in S} \frac{Y_i (d_i + \eps \Delta)}{p_i\gs} \le 2 n\Delta$.
\end{enumerate}
\end{obs}

All the inequalities are simple consequences of Bernstein's inequality (and our choice of parameters $p_i$, $p_i\gs$, $q_i$), and we thus skip the proofs. Next, note that as an immediate consequence of part-1, we have
\begin{equation}\label{eq:rho-bound-noise}
\rho_i \le 2(d_i +\eps \Delta), \qquad \text{for all partitions $(A, S\setminus A)$ of $S$}.
\end{equation}
Also, note that from the definitions of the quantities (and the fact $q_i/p_i \le \eps^2$), we have
\begin{equation} \label{eq:ppstar}
\forall i \not\in S,~~ \left| \frac{1}{p_i} - \frac{1}{p_i\gs} \right| \le \frac{\eps^2}{p_i\gs}
\end{equation}

Now, we are ready to show~\eqref{eq:to-prove-error}.  The first term can be bounded as follows:
\begin{equation}
\sum_i Y_i |\rhotil_i - \rhotil_i\gs| = \sum_i  \frac{Y_i \rho_i}{\Delta^2} \left| \frac{1}{p_i} - \frac{1}{p_i\gs} \right| = \sum_{i \in S} \frac{\rho_i}{\Delta^2} \left| \frac{1}{p_i} - 1 \right| + \sum_{i \not\in S} \frac{Y_i \rho_i}{\Delta^2} \left| \frac{1}{p_i} - \frac{1}{p_i\gs} \right|.\label{eq:noise2}
\end{equation}

Using~\eqref{eq:rho-bound-noise} and part-3 of the observation, the first term can be bounded by $O(\eps^2 n/\Delta)$.  For the second term, using~\eqref{eq:ppstar} together with part-4 of the observation gives a bound of $O(\eps^2 n/\Delta)$.  Thus the RHS above is at most $\frac{\eps n}{16 \Delta}$, as we may assume $\eps $ is small enough.

Now, consider the second term in~\eqref{eq:to-prove-error}.  When $i \not\in S$ and $j \not\in S$, we have $|w_{ij} - w_{ij}\gs|$ being ``small''. We can easily bound by $2 \eps^2 w_{ij}\gs$, using 
\[ \left| \frac{1}{p_i p_j} - \frac{1}{p_i\gs p_j\gs} \right|  \le \left| \frac{1}{p_i p_j} - \frac{1}{p_i\gs p_j} \right| + \left| \frac{1}{p_i\gs p_j} - \frac{1}{p_i\gs p_j\gs} \right| \le \frac{\eps^2}{p_i\gs p_j} + \frac{\eps^2}{p_i\gs p_j\gs}  \le \frac{2 \eps^2}{p_i\gs p_j\gs} . \]
In the last steps, we used~\eqref{eq:ppstar} and the fact that $p_i\gs \le p_i$ for $i \not\in S$.

For $i \in S$ or $j \in S$, we can simply bound $|w_{ij} - w_{ij}\gs|$ by $2 w_{ij}$.  Thus we can bound the second term in~\eqref{eq:to-prove-error} as
\[ 4 \sum_{i \in S} \frac{1}{p_i \Delta^2}\sum_{j \in \Gamma(i)} \frac{Y_j}{p_j} + \sum_{i \not\in S} \frac{Y_i}{p_i\gs \Delta^2} \sum_{j \in \Gamma(i) \setminus S} \frac{\eps^2 Y_j}{p_j}. \]
The second term has only a sum over $j$ not in $S$ -- this is why have an extra 2 factor for the first term. Now, consider the first term. The inner summation can be written as $\sum_{j \in \Gamma(i) \cap S} \frac{1}{p_j} + \sum_{j \in \Gamma(i) \setminus S} \frac{Y_j}{p_j} $. Using parts 1 and 2 of the observation, together with $p_j \ge q_j/\alpha_\eps$, and $p_j \ge p_j\gs$ for $j \not\in S$, we have $\sum_{j \in \Gamma(i)} \frac{Y_j}{p_j} \le  4 \alpha_\eps (d_i + \eps \Delta)$.  Then, using part-3 gives the desired bound on the first term.

Let us thus consider the second term.  Again using part 2 along with $p_j \ge p_j\gs$ for $j \not\in S$, we can bound the inner sum by $2\eps^2 (d_i +\eps \Delta)$.  Then we can appeal to part-4 of the observation to obtain the final claim.

This ends up bounding the second term of~\eqref{eq:to-prove-error}, thus completing the proof of the lemma.
\end{proof}

Thus our goal is to show the following:
\begin{lemma}\label{lem:helper-maxcut}
Let $S$ satisfy the conditions (a) $|S| \le \frac{20 n\log n}{\eps^2 \Delta}$, and (b) for all $i \in V$, $\sum_{j \in \Gamma(i) \cap S} \frac{1}{q_j} \le 2(d_i + \eps \Delta)$. Then, w.p. $\ge 1-\frac{1}{n^4}$ over the choice of $S'$ given $S$, we have
\[ \max_{A \subseteq S} \dual_H\gs \le \frac{1}{\Delta^2} \cdot \max_{A \subseteq S} \dual_G + \frac{\eps n}{2\Delta}. \]
\end{lemma}

The condition (b) on $S$ is a technical one that lets us bound $\rho_i$ in the proofs.  

Given Lemma \ref{lem:dual-error} and \ref{lem:helper-maxcut}, it is easy to prove the Theorem~\ref{thm:main-lp-sample} as follows:
\begin{proof}[Proof of Theorem~\ref{thm:main-lp-sample}]
The conditions we assumed on $S$ in Lemma~\ref{lem:helper-maxcut} hold w.p. at least $1-\frac{1}{n^4}$ (via a simple application of Bernstein's inequality). Thus the conclusion of the lemma holds w.p. at least $1-\frac{2}{n^4}$.  Combining this with Lemma~\ref{lem:dual-error}, we have that 
$\max_A \dual_H \le \frac{1}{\Delta^2}\max_A \dual_G + \frac{\eps n}{\Delta}$ w.p. at least $1-\frac{3}{n^4}$. The theorem then follows via LP duality.
\end{proof}

It thus suffices to prove Lemma~\ref{lem:helper-maxcut}.  The main step is to show a concentration bound on a quadratic function that is not quite a quadratic {\em form}.  This turns out to be quite technical, and we discuss it in the following sections.

\subsection{Proof of Lemma~\ref{lem:helper-maxcut}}\label{app:lem-helper-maxcut}

Let $Y_i = \mathbf{1}_{i \in S'}$. For convenience, let us denote the $\max\{\}$ terms in equations~\eqref{eq:dual-g} and~\eqref{eq:dual-h-gs} by $u_i$ and $\tilde{u}_i\gs$, respectively. Now, 
\begin{equation}
\label{eq:dual-diff}
\dual_H\gs - \frac{1}{\Delta^2} \dual_G = \sum_i \left( Y_i \rhotil_i\gs z_i - \frac{1}{\Delta^2} \cdot \rho_i z_i \right) + \left( Y_i \tilde{u}_i\gs - \frac{1}{\Delta^2} \cdot u_i \right).
\end{equation} 
We view the RHS as two summations (shown by the parentheses), and bound them separately. 

The first is relatively easy. Recall that by definition,  $\rhotil_i\gs = \frac{\rho_i}{p_i\gs \Delta^2}$.
Thus the first term is equal to $\sum_i \frac{\rho_i z_i}{p_i\gs \Delta^2} \big( Y_i  - p_i\gs \big)$. The expectation of this quantity is $0$.  We will apply Bernstein's inequality to bound its magnitude. For this, note that the variance is at most (using $|z_i| \le 1$)
\[ \sum_i \frac{\rho_i^2}{(p_i\gs)^2 \Delta^4} p_i\gs (1-p_i\gs) \le \sum_i \frac{4 (d_i +\eps \Delta)^2 (1-p_i\gs)}{p_i\gs \Delta^4}. \]
The condition on $S$ gives the bound on $\rho_i$ that was used above. Next, we note that unless $p_i\gs = 1$, we have $p_i\gs \ge \frac{(d_i + \eps \Delta)}{\alpha_\eps \Delta^2}$. Thus the variance is bounded by $\sum_i \frac{4\alpha_\eps \cdot (d_i + \eps \Delta)}{\Delta^2} \le \frac{8 \alpha_\eps n}{\Delta}$. Next, 
\[ \max_i \left| \frac{\rho_i z_i}{p_i\gs \Delta^2} \right| \le \frac{2(d_i +\eps \Delta)}{p_i\gs \Delta^2} \le 2\alpha_\eps .\]
(Again, this is because we can ignore terms with $p_i\gs = 1$, and for the rest, we have a lower bound.) Thus, by Bernstein's inequality, 
\[ \Pr[ \left|\sum_i \frac{\rho_i z_i}{p_i\gs \Delta^2} \big( Y_i  - p_i\gs \big)\right| \ge t ] \le  \exp \left( -\frac{ t^2}{\frac{16 \alpha_\eps n}{\Delta} + 2t\alpha_\eps} \right).\]
Setting $t = \eps n/4\Delta$, the bound simplifies to $\exp( - \frac{\eps^2 n}{C\Delta \alpha_\eps})$, for a constant $C$. Thus, by our choice of $\alpha_\eps$ and our size bound on $|S|$, this is $< \exp(-|S|)/n^4$.

The second term of~\eqref{eq:dual-diff} requires most of the work. We start with the trick (which turns out to be important) of splitting it into two terms by adding a ``hybrid'' term, as follows:
\[ \sum_i Y_i \tilde{u}_i\gs - \frac{1}{\Delta^2} \cdot u_i = \sum_i \left( Y_i \tilde{u}_i\gs - Y_i \frac{u_i}{p_i\gs \Delta^2}\right) + \sum_i \left( Y_i \frac{u_i}{p_i\gs \Delta^2} - \frac{1}{\Delta^2} \cdot u_i\right). \]

The second term will again be bounded using Bernstein's inequality (in which we use our earlier observation that $u_i = O(d_i)$). This gives an upper bound of $\eps n/8\Delta$, with probability $1- \exp(-|S|)/n^4$.  We omit the easy details.  

Let us focus on the first term. We now use the simple observation that $\max\{0, A\} - \max\{0, B\} \le |A- B|$, to bound it by
\[ \sum_i Y_i \left| \dtil_i\gs - \rhotil_i\gs -\sum_{j \in \Gamma(i) \cap S'} w_{ij}\gs z_j - \frac{1}{p_i\gs \Delta^2} \big( d_i - \rho_i - \sum_{j \in \Gamma(i)} {z_j} \big) \right|. \]

By the definition of $\rhotil_i\gs$, it cancels out.  Now, writing $c_j = 1-z_j$ (which now $\in [0,2]$) and using the definition of $\dtil_i\gs$, we can bound the above by
\[ \sum_i Y_i \left| \sum_{j \in \Gamma(i)} Y_j w_{ij}\gs c_j - \frac{1}{p_i\gs \Delta^2} c_j \right| = \sum_i Y_i \left| \sum_{j \in \Gamma(i)} w_{ij}\gs c_j (Y_j - p_j\gs) \right| \qquad \text{(using $w_{ij}\gs = \frac{1}{p_i\gs p_j\gs \Delta^2}$)} \]

Showing a concentration bound for such a quadratic function will be subject of the rest of the section. Let us define
\begin{equation}\label{eq:def-f}
f(Y) := f(Y_1, \dots, Y_n) := \sum_i Y_i \left| \sum_{j \in \Gamma(i)} w_{ij}\gs c_j (Y_j - p_j\gs) \right|.
\end{equation}

We wish to show that $\Pr[ f > \frac{\eps n}{\Delta} ] \le \exp(-|S|)$. Unfortunately, this is not true -- there are counter-examples (in which the neighborhoods of vertices have significant overlaps) for which it is not possible to obtain a tail bound better than $\exp(-n/\Delta^2)$, roughly speaking. To remedy this, we resort to a trick developed in~\cite{goldreich1998property, feige2002optimality}. The key idea is to {\em condition} on the event that vertices have a small weighted degree into the set $S'$, and obtain a stronger tail bound.

\npara{``Good'' conditioning.} We say that a choice of $Y$'s is {\em good} if for all $i \in V$, we have 
\[ \sum_{j \in \Gamma(i)} w_{ij}\gs Y_j \le \frac{\eps\Delta + 2d_i}{p_i\gs \Delta^2}. \]

The first lemma is the following.

\begin{lemma}\label{lem:allgood}
	Let $H$ be the weighted graph on $S'$ obtained by our algorithm. For any vertex $i \in V$, we have 
	\[ \Pr \left[ \sum_{j \in \Gamma(i)} w_{ij}\gs Y_j > \frac{\eps\Delta + 2d_i}{p_i\gs \Delta^2} \right] < \frac{1}{n^6}. \]
\end{lemma}
\begin{proof}
	Fix some $i \in V$, and consider $\sum_{j \in \Gamma(i)} w_{ij}\gs Y_j = \frac{1}{p_i\gs \Delta^2} \left( \sum_{j \in \Gamma(i)} \frac{Y_j}{p_j\gs} \right)$. The term in the parenthesis has expectation precisely $d_i$. Thus, applying Bernstein using $\max_j \frac{1}{p_j\gs} \le \frac{\alpha_\eps \Delta}{\eps}$, together with $\sum_{j \in \Gamma(i)} \frac{p_j\gs(1-p_j\gs)}{(p_j\gs)^2} \le d_i \max_j \frac{1}{p_j\gs}$, we have
	\[  \Pr \big[ \sum_{j \in \Gamma(i) \cap V_H} \frac{Y_j}{p_j\gs} > d_i + t \big] \le \exp \left( - \frac{\eps t^2}{(d_i + t) \alpha_\eps \Delta} \right). \]
	By setting $t = (d_i + \eps \Delta)$, the RHS above can be bounded by
	\[ \exp \left( - \frac{\eps (d_i + \eps \Delta)^2}{(2d_i + \eps \Delta) \alpha_\eps \Delta} \right) \le \exp\left( - \frac{\eps^2}{2 \alpha_\eps} \right) < \frac{1}{n^6}.\]
	
	This completes the proof, using our choice of $\alpha_\eps$.
\end{proof}

Conditioning on the $Y$ being good, we show the following concentration theorem.
\begin{theorem}\label{thm:conc:condition}
	Let $Y_i$'s be independent random variables, that are $1$ w.p. $p_i\gs$ and $0$ otherwise, and let $f(Y)$ be defined as in~\eqref{eq:def-f}. Then we have
	\[ \Pr \big[ f(Y) \ge \frac{\eps n}{8 \Delta} ~\big|~ Y \text{ is good} \big] \le \frac{1}{n^5} \cdot e^{-20n \log n/\eps^2}. \]
\end{theorem}

We observe that the theorem implies Lemma~\ref{lem:helper-maxcut}. This is because by the Theorem and the preceeding discussions, $\Pr[ \dual_H\gs - \dual_G \le \eps n/\Delta~|~ Y \text{ good} ] \ge 1- \frac{\exp(-|S|)}{n^5}$, for any $A \subseteq S$.  Thus by union bound over $A$, $ \Pr[ \max_{A} \dual_H\gs - \max_{A} \dual_G \le \eps n/\Delta~|~ Y \text{ good} ] \ge 1-\frac{1}{n^5}$.  Since the probability of the good event is at least $1-\frac{1}{n^5}$ (by Lemma~\ref{lem:allgood}), the desired conclusion follows.


\subsection{Concentration bound for quadratic functions}
\label{sec:blm}

To conclude our proof, it suffices to show Theorem~\ref{thm:conc:condition}.  To bound the quadratic function $f$, we bound the moment generating function (MGF), $\E[ e^{\lambda f}~|~\text{good} ]$. This is done via a decoupling argument, a standard tool for dealing with quadratic functions. While decoupling is immediate for `standard' quadratic forms, the proof also works for our $f$ (which has additional absolute values). The rest of the proof has the following outline.

\noindent {\em Proof outline.} The main challenge is the computation of the MGF under conditioning (which introduces dependencies among the $Y_i$, albeit mild ones). The decoupling allows us to partition vertices into two sets, and only consider edges that go across the sets. We then show that it suffices to bound the MGF under a ``weakened'' notion of conditioning (a property we call $\delta$-good). Under this condition, all the vertices in one of the sets of the partition become independent, thus allowing a bound on the moment --- in terms of quantities that depend on the variables on the other set of the partition. Finally, we appeal to a strong concentration bound of Boucheron et al.~\cite{boucheron2003concentration} to obtain an overall bound, completing the proof.



We now expand the proof outline above.

\npara{Decoupling.}
Consider independent Bernoulli random variables $\delta_i$ that take values $0$ and $1$ w.p. 1/2 each, and consider the function
\[ f_\delta := \sum_{i} \delta_i  Y_i  | \sum_{j \in \Gamma(i)} (1- \delta_j) w_{ij}\gs c_j (Y_j -p_j\gs) | \]
Using the fact that $\mathbb{E}[|g(x)|] \geq |\mathbb{E}[g(x)]|$ for any function $g$, and defining $\mathbb{E}_\delta$ as the expectation with respect to the $\delta_i$'s, we have 

\begin{align*}
\E_\delta f_\delta &= \E_\delta \sum_{i} \delta_i  Y_i  | \sum_{j \in \Gamma(i)} (1- \delta_j) w_{ij}\gs c_j (Y_j -p_j\gs)| 
=  \sum_{i} \frac{1}{2} \cdot Y_i  \E_\delta | \sum_{j \in \Gamma(i)} (1- \delta_j) w_{ij}\gs c_j (Y_j -p_j\gs)|  \\
&\geq \sum_{i} \frac{1}{2} \cdot Y_i | \sum_{j \in \Gamma(i)} \E_\delta (1- \delta_j) w_{ij}\gs c_j (Y_j -p_j\gs)| \\
&= \sum_{i} \frac{1}{4} \cdot Y_i | \sum_{j \in \Gamma(i)}  w_{ij}\gs c_j (Y_j -p_j\gs)| = \frac{1}{4} f
\end{align*}

We used the fact that $i$ never appears in the summation term involving $Y_i$ to obtain the first equality. Next, using Jensen's inequality, we have:
\begin{align*}
\E_Y [e^{\lambda f}~|~\text{good}] \le \E_Y[ e^{4 \lambda \E_\delta f_\delta}~|~\text{good}] \le \mathbb{E}_{Y,\delta} [e^{4 \lambda f_\delta}~|~\text{good}]
\end{align*}
where $\mathbb{E}_{Y, \delta}$ means the expectation with respect to both random variables $Y$ and $\delta$. Now, the {\em interpretation} of $f_\delta$ is simply the following. Consider the partitioning $(V^+, V^-)$ of $V$ defined by $V^+ = \{i \in [n]: \delta_i =1\} $ and $V^- = \{i \in [n]: \delta_i =0\}$, then
\begin{align*}
f_\delta = \sum_{i \in V^+} Y_i \big| \sum_{j \in \Gamma(i)\cap V^-} w_{ij}\gs c_j (Y_j -p_j\gs) \big|.
\end{align*}

For convenience, define $R_i = | \sum_{j \in \Gamma(i)\cap V^-} w_{ij}\gs c_j (Y_j -p_j\gs)|$, for $i \in V^+$. Thus we can write $f_\delta = \sum_{i \in V^+} Y_i  R_i$. The condition that $Y$ is good now gives us a bound on $R_i$. For any $c_j$ (it is important to note that the good condition does not involve the constants $c_j$, as those depend on the LP solution; all we know is that $0 \le c_j \le 2$), we have
\begin{align*}
R_i &\le \big| \sum_{j \in \Gamma(i) \cap V^-} 2 w_{ij}\gs (Y_j + p_j\gs) \big| \\
&\le \frac{2(\eps \Delta + 2d_i)}{p_i\gs \Delta^2} + \frac{2d_i}{p_i\gs \Delta^2} \le \frac{2 \eps \Delta + 6 d_i}{p_i\gs \Delta^2}.
\end{align*}

Now, the quantity we wish to bound can be written as 
\begin{equation}\label{eq:moment-tobound} 
\Exp_Y [e^{\lambda f}~|~\text{good}] \le \Exp_{Y,\delta}~[e^{4 \lambda f_\delta}~|~\text{good}] \le \Exp_{\delta} \Exp_{Y^-} \Exp_{Y^+}~[e^{4 \lambda \sum_{i \in V^+} Y_i R_i}~|~\text{good}].
\end{equation}
The key advantage that decoupling gives us is that we can now {\em integrate over} $Y_i \in V^+$, i.e., evaluate the innermost expectation, for any given choice of $\{Y_i: i \in V^-\}$ (which define the $R_i$). The problem with doing this in our case is that the good condition introduces dependencies on the $Y_i$, for $i \in V^+$.

Fortunately, weakening conditioning does not hurt much in computing expectations. This is captured by the following simple lemma. 
\begin{lemma} \label{lem:condition}
Let $\Omega$ be a space with a probability measure $\mu$. Let $Q_1$ and $Q_2$ be any two events such that $Q_1 \subset Q_2$, and let $Z : \Omega \mapsto \mathbb{R^+}$ be a {\em non-negative} random variable. Then,
\begin{align*}
\E[Z | Q_1] \leq \frac{\E[Z|Q_2]}{\Pr[Q_1]}.
\end{align*}
\end{lemma}
\begin{proof}
Let $\Omega_1$ (resp. $\Omega_2$) be the subset of $\Omega$ in which $Q_1$ (resp. $Q_2$) is satisfied. By hypothesis, $\Omega_1 \subseteq \Omega_2$. Now by the definition of conditional expectation, and the non-negativity of $Z$, we have
\[ \E[ X | Q_1 ] = \frac{1}{\mu(Q_1)} \int_{x \in \Omega_1} Z(x) \mu(x) dx \le \frac{1}{\mu(Q_1)} \int_{x \in \Omega_2} Z(x) \mu(x) dx = \frac{\mu(Q_2)}{\mu(Q_1)} \E[ Z | Q_2]. \]
Since $\mu(Q_2) \le 1$, the conclusion follows.
\end{proof}

\npara{Weaker {\em good} property.} The next crucial notion we define is a property ``$\delta$-good''. Given a $\delta \in \{0,1\}^n$ (and corresponding partition $(V^+, V^-)$), a set of random variables $Y$ is said to be $\delta$-good if for all $i \in V^+$, we have
\begin{equation}\label{eq:def-ri}
\sum_{j \in \Gamma(i) \cap V^-} w_{ij}\gs Y_j \le \frac{\eps\Delta + 2d_i}{p_i\gs \Delta^2}.
\end{equation}

We make two observations. First, the good property implies the $\delta$-good property, for any choice of $\delta$.  Second, and more crucial to our proof, conditioning on $\delta$-good does not introduce any dependencies on the variables $\{ Y_i:~i \in V^+\}$.  Now, continuing from~\eqref{eq:moment-tobound}, and using the fact that the good condition holds with probability $> 1/2$, we have
\[ \Exp_{\delta} \Exp_{Y^-} \Exp_{Y^+}~[e^{4 \lambda \sum_{i \in V^+} Y_i R_i}~|~\text{good}] \le \Exp_{\delta} \Exp_{Y^-} \Exp_{Y^+}~[2 e^{4 \lambda \sum_{i \in V^+} Y_i R_i}~|~\text{$\delta$-good}]. \]

Now for any $0/1$ choices for variables $Y^-$, the $R_i$'s get fixed for every $i \in V^+$, and we can bound $\E_{Y^+} [ e^{4\lambda \sum_{i \in V^+} Y_i R_i}~|~R_i ]$ easily. 
\begin{lemma}\label{lem:wt-bernoulli}
Let $Y_i$ be independent random $0/1$ variables taking value $1$ w.p. $p_i\gs$, and let $R_i$ be given, for $i \in V^+$. Suppose $\lambda >0$ satisfies $|\lambda R_i| \le 1$ for all $i$. Then
\[ \E_{Y^+} [ e^{\lambda \sum_{i \in V^+} Y_i R_i} ] \le e^{\sum_{i \in V^+} \lambda p_i\gs R_i + \lambda^2 p_i\gs R_i^2}  .\]
\end{lemma}
\begin{proof}
Since the lemma only deals with $i \in V^+$, we drop the subscript for the summations and expectations. One simple fact we use is that for a random variable $Z$ with $|Z|\le 1$,
\[ \E[e^Z] \le \E[ 1+Z+Z^2 ] \le e^{\E[Z] + \E[Z^2]}.\]
Using this, and the independence of $Y_i$ together with $Y_i^2 = Y_i$,
\begin{equation}
\E[ e^{\lambda \sum Y_i R_i} ] = \prod \E [ e^{\lambda Y_i R_i} ] \le \prod e^{\E[\lambda Y_i R_i] + \E[\lambda^2 R_i^2 Y_i]}.
\end{equation}
As $\E[Y_i] = p_i\gs$, this completes the proof of the lemma.
\end{proof}

Using the lemma, replacing $\lambda$ with $4\lambda$ yields the following
\begin{equation}
\Exp_{\delta} \Exp_{Y^-} \Exp_{Y^+}~[e^{4 \lambda \sum_{i \in V^+} Y_i R_i}~|~\text{$\delta$-good}] \le \Exp_{\delta} \Exp_{Y^-} \big[ e^{\sum_{i \in V^+} 4 \lambda p_i\gs R_i + 16 \lambda^2 p_i\gs R_i^2}~|~\text{$\delta$-good} \big].
\end{equation}

The second term in the summation is already small enough. I.e., using \eqref{eq:def-ri}
\[ \sum_{i \in V^+} \lambda^2 p_i\gs R_i^2 \le 2 \lambda^2 \alpha_\eps \frac{n}{\Delta}. \]

While the bound on $R_i$ can be used to bound the first term, it turns out that this is not good enough. We thus need a more involved argument. Thus the focus is now to bound
\[ \Exp_\delta \Exp_{Y^-} \big[ e^{\lambda g(Y)}~|~ \text{$\delta$-good} \big], \text{ where } g(Y) := \sum_{i \in V^+} p_i\gs \left| \sum_{j \in \Gamma(i) \cap V^-} w_{ij}\gs c_j (Y_j -p_j\gs) \right|. \]

\npara{Outline: concentration bound for $g$.}  To deduce a concentration bound for $g$, we first remove the conditioning (again appealing to Lemma~\ref{lem:condition}). This then gives us independence for the $Y_j$, for $j \in V^-$. We can then appeal to the fact that changing a $Y_j$ only changes $g$ by a small amount, to argue concentration. However, the standard ``bounded differences'' concentration bound (\cite{boucheron2003concentration}, Theorem ...) will not suffice for our purpose, and we need more sophisticated results~\cite{boucheron2003concentration} (restated as Theorem~\ref{concentration-main}).  

To use the same notation as the theorem, define $Z = g(Y)$, where we only consider $Y_r$, $r \in V^-$. Now, for any such $r$, consider $Z - Z\suu{r}$. Since $Z\suu{r}$ is obtained by replacing $Y_r$ by an independent $Y_r'$ and re-computing $g$, we can see that the only terms $i$ which could possibly be affected are $i \in \Gamma(r) \cap V^+$. Further, we can bound the difference $|Z - Z\suu{r}|$ by
\[|Z- Z^{(r)}| \le |Y_r - Y_r'| \sum _{i \in \Gamma(r) \cap V^+} 2p_i\gs w_{ir}\gs,\]
where we have used $|c_r| \le 2$. The summation can be bounded by $\frac{2 d_r}{p_r\gs \Delta^2}$. Denote this quantity by $\theta_r$. Then, to use the theorem, we need
\[ \E_{Y_r'} |Z - Z\suu{r}|^2 \le |Y_r - Y_r'|^2 \theta_r^2 \le |Y_r - Y_r'| \theta_r^2 \le Y_r \theta_r^2 + p_r\gs \theta_r^2. \]
We used the fact that $|Y_r - Y_r'| \in \{0,1\}$. Now applying Theorem~\ref{concentration-main} by setting $\theta = \frac{1}{2 \lambda}$, we have:
\[ \mathbb{E}_{V^-}[e^{\lambda (g - \E[g])}] \leq \mathbb{E}_{V^-}[e^{\frac{\lambda^2}{2}\sum_{r \in V^-} Y_r \theta_r^2 + p_r\gs \theta_r^2}] \]
Once again, since $\lambda \theta_r$ will turn out to be $<1$, we can use the bound $\E[ e^{\lambda^2 Y_r \theta_r^2/2}] \le e^{\lambda^2 \theta_r^2 p_r\gs}$, and conclude that
\[ \mathbb{E}_{V^-}[e^{\lambda (g - \E[g])}] \le e^{2 \lambda^2 \sum_{r \in V^-} p_r\gs \theta_r^2} \le e^{4\lambda^2 \alpha_\eps \frac{n}{\Delta}}.\]

The last inequality is due to a reasoning similar to earlier.

We are nearly done. The only step that remains for proving Theorem~\ref{thm:conc:condition} is to obtain a bound on $\E[g]$. For this, we need to bound, for any $i$, the term 
\[ \E[ |R_i| ] = \E \huge[ \large| \sum_{j \in \Gamma(i) \cap V^-} w_{ij}\gs c_j (Y_j - p_j\gs) \large| \huge].\]
By Cauchy-Schwartz and the fact that $Y_j$ are independent, we have
\begin{align*}
\E[ |R_i| ]^2 &\le \E[ R_i^2 ] = \sum_{j \in \Gamma(i) \cap V^-} (w_{ij}\gs)^2 c_j p_j\gs (1-p_j\gs) \\
&\le \frac{1}{(p_i\gs)^2 \Delta^4} \sum_{j \in \Gamma(i)} \frac{1-p_j\gs}{p_j\gs} \le \frac{\alpha_\eps d_i}{\eps (p_i\gs)^2 \Delta^3}.
\end{align*}
We have used the fact that $(1-p_j\gs)/p_j\gs \le \alpha_\eps \Delta/ \eps$. 
Thus we have
\[ \E[ \sum_{i \in V^+} p_i\gs |R_i| ] \le \left(\frac{\alpha_\eps}{\eps \Delta^3} \right)^{1/2} \sum_i d_i^{1/2} \le \left(\frac{\alpha_\eps}{\eps \Delta^3} \right)^{1/2} n \Delta^{1/2} \le \frac{n}{\Delta} \left(\frac{\alpha_\eps}{\eps} \right)^{1/2}. \]
From our choice of $\alpha_\eps$, this is $< \frac14 \cdot \frac{\eps n}{\Delta}$.

Putting everything together, we get that the desired moment 
\[ \le \exp \left( \lambda \frac{\eps n}{4\Delta} + \lambda^2 \alpha_\eps \frac{n}{\Delta} \right). \]

To complete the bound, we end up setting $\lambda = \frac{\eps}{\alpha_\eps}$.  For this value of $\lambda$, we must ensure that 
\[  \frac{\eps}{\alpha_\eps}  \frac{(d_r + \eps \Delta)}{p_r\gs \Delta^2} \le 1, \]
which is indeed true.

%
%

\section{Sparse Core-set for Max-Cut}
\label{sec:edge-sample}
In Theorem~\ref{sampling}, we have shown that there is a core-set (i.e., a smaller weighted graph with the same MAXCUT value) with a small number of {\em vertices}.  We now show that the number of edges can also be made small (roughly $n/\Delta$).  This will prove Theorem~\ref{coreset}.

We start with a lemma about sampling edges in graphs with edge weights $\le 1$. (We note that essentially the same lemma is used in several works on sparsifiers for cuts.)
\begin{lemma}\label{edge-sample}
Consider the weighted graph $H$ resulted from \vertexsample, defined on set of vertices $S'$ in which $w_{ij} < 1$. If we apply \edgesample on $H$, then the resulted graph $H'$ have 
\[ \maxcut (H') \pm (1+\epsilon) \maxcut (H)\]
with probability at least $1- \frac{1}{n^2}$.
\end{lemma}

Let us first see that the lemma gives us Theorem~\ref{coreset}.
\begin{proof}[Proof of Theorem~\ref{coreset}]
We only need to verify the bound on the number of edges. As every edge is sampled with probability $p_{ij} = \min(1, \frac{8 w_{ij}}{\eps^2})$, and since the total edge weight is normalized to be $|S'|$, we have that the expected number of edges is $\le \frac{8|S'| }{\eps^2}$, and w.p. at least $1-\frac{1}{n^4}$, this is at most $\frac{16 |S'| }{\eps^2}$, completing the proof.
\end{proof}


\subsection{Proof of Edge Sampling}\label{app:edge-sample}
We now prove Lemma~\ref{edge-sample}.

\begin{proof}
Recall that in \edgesample algorithm we first rescale the edge weights so that they sum up to $|S'|$, and then sample each edge $ij$ in the resulted graph (also denoted $H$, as we can assume it to be a pre-processing) with probability $p_{ij} = \min (w_{ij} \frac{8}{\epsilon^2}, 1)$ and reweigh the edges to $w'_{ij} = \frac{w_{ij}}{p_{ij}}$ to obtain the graph $H'$. Define the indicator variable $X_{ij}$ for each edge $e_{ij}$ in graph $H$, where $X_{ij} =1$ if the corresponding edge $e_{ij}$ is selected by \edgesample and $X_{ij} = 0$ otherwise. 

Our goal is to show that all cuts are preserved, w.h.p.  Consider any cut $(A, B)$ in $H$.  Call the set of all the edges on this cut as $C_{A, B}$ and the set of in the cut on the sampled graph $H'$ as $C'_{A, B}$. 
Set $w(C_{A, B}) = \sum_{e_{ij} \in C_{A, B}} w_{ij}$ and $w'(C'_{A, B}) = \sum_{e_{ij} \in C'_{A, B}} w'_{ij}$. Then we have,
\[ \mathbb{E}[w'(C'_{A, B})] = \mathbb{E}[\sum w'_{ij} X_{ij}] = \sum w'_{ij} \Pr (X_{ij} = 1) = \sum p_{ij} w'_{ij} = \sum w_{ij} = w(C_{A, B}).\]

We will now apply Bernstein's inequality to bound the deviation. For this, the variance is first bounded as follows.
\[ \Var[\sum w'_{ij} X_{ij}] = \sum {w'_{ij}}^2 \Var (X_{ij}) \leq  \sum \frac{w_{ij}^2}{p_{ij}} (1-p_{ij}) \le \frac{\epsilon^2}{8} \sum w_{ij} = \frac{\epsilon^2}{8} w(C_{A, B})\]
We used the inequality that unless $p_{ij}=1$ (in which case the term drops out), we have $w_{ij}/p_{ij} \le \eps^2/8$.  

By the observation on $w_{ij}/p_{ij}$ above, we can use Bernstein's inequality \ref{bern} with $b = \eps^2/8$, to obtain
\[ 
\Pr[|\sum w'_{ij} X_{ij} - w(C_{A, B})| \ge t ] \le \exp \left( - \frac{t^2}{ \frac{\eps^2 w(C_{A, B})}{4} + \frac{t \eps^2}{8} } \right). 
\]

Setting $t = \eps W$, where $W$ is the sum of all the edge weights (which is equal to $|S'|$ after the pre-processing), the bound above simplifies to $\exp(-2|S'|)$, and thus we can take a union bound over all cuts. This completes the proof.
\end{proof}

\section{A 2-pass streaming algorithm}
\label{sec:2-pass-streaming}

We now show how our main core set result can be used to design a streaming
algorithm for \maxcut. The algorithm works in two passes: the first pass builds
a core-set $S$ of size $\otilde(n/\Delta)$ as prescribed by
Theorem~\ref{sampling} and the second pass builds the induced weighted graph
$G[S]$ and computes its max cut. This algorithm works under edge insertion/deletion.

\subsection{Pass 1: Building a core set}
\label{sec:pass-1:-building}

To construct $S$, Theorem~\ref{sampling} states that each vertex must be sampled
with probability $p_i$, where $p_i \ge h_i$ and
$h_i = \min(1, \frac{\max(d_i, \epsilon \Delta)}{\Delta^2\alpha_\epsilon})$ is
the importance score of a vertex. As the goal is to only choose a small number
of vertices, we will also make sure that $p_i \le 2 h_i$.  The challenge here is
two-fold: we need to sample (roughly) proportional to the degree $d_i$, which
can only be computed \emph{after} the stream has passed, and we also need the
actual value of $p_i$ (or a close enough estimate of it) in order to correctly
reweight edges in the second pass.

The degree $d_i$ of a vertex $i$ is the ``count'' of the number of times $i$ appears in the edge stream. To sample with probability proportional to $d_i$ we will therefore make use of streaming algorithms for $\ell_1$-sampling \cite{monemizadeh20101, andoni2011streaming, jowhari2011tight}. We borrow some notation from \cite{andoni2011streaming}. 

\begin{definition}
 Let $\rho > 0, f \in [1,2]$. A $(\rho, f)$-approximator to $\tau > 0$ is a quantity $\hat{\tau}$ such that 
$ \tau/f - \rho \le \hat{\tau} \le f \tau + \rho $
\end{definition}


\begin{lemma}[\cite{jowhari2011tight} (rephrased from \cite{andoni2011streaming})]
\label{lemma:jst}
Given a vector $x \in \reals^n$ and parameters $\epsilon, \delta > 0, c > 0$ there exists an algorithm $A$ that uses space $O(\log(1/\epsilon)\epsilon^{-1} \log^2 n\log(1/\delta))$ and generates a pair $(i, v)$ from a distribution $D_x$ on $[1\dots n]$ such that with probability $1-\delta$
\begin{itemize}
\item $D_x(i)$ is a $(\frac{1}{n^c}, 1+\epsilon)$-approximator to $|x_i|/\|x\|_1$
\item $v$ is a $(0, 1+\epsilon)$-approximator to $x_i$
\end{itemize}
where $c$ is a fixed constant.
\end{lemma}

We will also need to maintain heavy hitters: all vertices of degree at least $\Delta^2$ (up to constants). To do this, we will make use of the standard \textsc{CountMin} sketch \cite{cormode2005improved}. For completeness, we state its properties here.

\begin{lemma}[\cite{cormode2005improved}]
\label{lemma:cm}
Fix parameters $k, \delta > 0$. Then given a stream of $m$ updates to a
vector $x \in \reals^n$ there is a sketch CM of size $O(k\log \delta^{-1}(\log m
+ \log n))$ and a reconstruction procedure $f:[n] \to \reals$ such that with
probability $1-\delta$, for any $x_i$,  
$ |x_i - f(i)| \le  \|x\|_1/k $
\end{lemma}

\noindent {\bf \emph{Outline.}}  We will have a collection of $r$, roughly $n/\Delta$ $\ell_1$-samplers. These samplers will together give a good estimate ($(1+\eps)$-approximation) of the importance $h_i$ for all the vertices that have a {\em small} degree (which we define to be $< \alpha_\eps \Delta^2$).  Then, we use the CM sketch to maintain the degrees of all the `high degree' vertices, i.e., those with degrees $\ge \alpha_\eps \Delta^2$. Taken together, we obtain the desired sampling in the first pass.  

\begin{definition}
  Given two sets of pairs of numbers $S, S'$, let 
$S \cup_{\max} S' = \{ (x, \max_{(x', y) \in S \cup S', x' = x} y) \}$
\end{definition}



\begin{algorithm}
\caption{Given $n$ and average degree $\Delta$\label{alg:pass1}}
  \begin{algorithmic}
    \STATE Initialize $S_l, S_m, S_h \leftarrow \emptyset$, and $\zeta =
    \alpha_\eps$.
    \STATE Sample elements from $[1\dots n]$ each with probability
    $\epsilon/\Delta\alpha_\epsilon$. For each sampled $i$ add $(i,
    \epsilon/\Delta\alpha_\epsilon)$ to $S_l$.
    \STATE Fix $\zeta > 0$. Initialize a \textsc{CountMin} sketch CM with size parameter $k =
    n/\Delta \zeta^2$. Let $f$ be the associated reconstruction procedure.
    \STATE Initialize $r = O(n/\Delta \alpha_\epsilon)$ copies $A_1\dots A_r$ of the algorithm $A$ from Lemma~\ref{lemma:jst}. 
    \FOR{each stream update $(i, w)$ (a vertex to which current edge is incident, and weight)}
    \STATE Update each $A_j, 1 \le j \le r$.
    \STATE Update CM. 
    \ENDFOR
    \FOR{$j = 1$ to $r$}
    \STATE Sample $(i, v)$  from $A_j$. $S_m = S_m \cup_{\max} \{(i, v)\}$
    \ENDFOR
    \STATE $S_h = \{ (i, 1) \mid f(i) \ge (1-\zeta)\Delta^2\alpha_\epsilon \}$ 
    \RETURN $S_l \cup_{\max} S_m \cup_{\max} S_h$
  \end{algorithmic}
\end{algorithm}

\begin{lemma}
  \label{lemma:pass1}
Let $S = \{ (i, v_i) \}$ be the set returned by Algorithm \ref{alg:pass1}. Then 
\begin{itemize}
\item $S$ has size $\widetilde{O}\left( \frac{n}{\Delta} \right)$.
\item Each $i \in [n]$ is sampled with probability $p_i$ that is $(0, 1+\epsilon)$ approximated by $v_i$ and that $(n^{-c}, 1+\epsilon)$-approximates $h_i$. 
\end{itemize}
\end{lemma}

\begin{proof}[Proof sketch.]
Consider any vertex $i$ with $d_i \ge \Delta^2 \alpha_\epsilon$. By Lemma \ref{lemma:cm}, such a
vertex will report a count of at least $f(i) =(1-\zeta)\Delta^2\alpha_\epsilon$ and thus is
guaranteed to be included in $S_h$. Its reported score $v_i = 1$ satisfies the
requirement of the Lemma. Secondly, consider any vertex with degree $d_i <
\epsilon\Delta$. For such a vertex, $h_i = \epsilon/\Delta\alpha_\epsilon$ and
thus it is included in $S_l$ with the desired probability and $v_i$.

Finally, consider a vertex $i$ with $\epsilon\Delta \le d_i < \Delta^2
\alpha_\epsilon$. 
The probability that none of the
$\ell_1$-samplers yield $i$ is $(1-d_i/n\Delta)^r$, and since $d_i/n\Delta \ll
1$, this can be approximated as $(1- r d_i/n\Delta)$. Thus, the probability of
seeing $i$ is $r d_i/n\Delta = d_i/\Delta^2 \alpha_\epsilon$ as desired.
\end{proof}

\begin{corollary}
\label{corollary:pass1}
For each $(i, v_i) \in S, h_i \le v_i \le 2h_i$. 
\end{corollary}

\subsection{Pass 2: Building the induced weighted graph}
\label{sec:pass-2:-building}

The first pass produces a set $S$ of $\otilde(n/\Delta)$ vertices together with estimates $v_i$ for their importance score $h_i$. If we weight each edge $ij$ in $G[S]$ by $w_{ij} = 1/v_iv_j \Delta^2$, Theorem \ref{sampling}, along with Corollary~\ref{corollary:pass1} guarantee that a \maxcut in the resulting weighted graph is a good approximation of the true max cut. 

Thus, knowing $S$, constructing the re-weighted $G[S]$ in the second pass is trivial if we had space proportional to the number of edges in $G[S]$. Unfortunately this can be quadratic in $|S|$, so our goal is to implement the edge sampling of Theorem~\ref{coreset} in the second pass. This is done as follows:
%
we maintain a set of edges $E'$. Every time
we encounter an edge $ij$ with $i, j \in S$, we check to see if it is already in $E'$. If not, we toss a coin and  with probability $p_{ij} = \min(1, w_{ij}\log n/\epsilon^2)$ we insert $(i,j,w_{ij}/p_{ij})$
into $E'$. By Lemma~\ref{edge-sample}, the size of $E'$ is $\widetilde{O}(n/\Delta)$, and
the resulting graph yields a $(1+\epsilon)$ approximation to the \maxcut.


\section{Correlation Clustering}\label{sec:correlation}

Our argument for correlation clustering parallels the one for \maxcut.  The MAX-AGREE variant of correlation clustering, while not a CSP (as the number of clusters is arbitrary), almost behaves as one.  

We start with two simple observations.  The first is that we can restrict the number of clusters to $1/\eps$, for the purposes of a $(1+\eps)$ approximation (Lemma~\ref{lem:corr-smallk}) .  Next, observe that the optimum objective value is at least $\max \{C^+, C^- \} \ge n\Delta/2$.  This is simply because placing all the vertices in a single cluster gives a value $C^+$, while placing them all in different clusters gives $C^-$.  Thus, it suffices to focus on additive approximation of $\eps n \Delta$. 

\begin{lemma}\label{lem:corr-smallk}
Let $\cC$ be the optimal clustering, and let OPT be its max-agree cost. Then there exists a clustering $\cC'$ that has cost $\ge (1-O(\eps)) OPT$, and has at most $1/\eps$ clusters. 
\end{lemma}

The lemma is folklore in the correlation clustering literature.
\begin{proof}
Let $A_1, A_2, \dots, A_k$ be the clusters in the optimal clustering $\cC$.   Now, suppose we randomly color the clusters with $t = 1/\eps$ colors, i.e., each cluster $A_i$ is colored with a random color in $[t]$. We then merge all the clusters of a given color into one cluster, thus obtaining the clustering $\cC'$. Clearly, $\cC'$ has at most $t$ colors.

Now, we observe that if $u, v \in A_i$ to begin with, then $u, v$ are still in the same cluster in $\cC'$. But if $u \in A_i$ and $v \in A_j$, and the colors of $A_i$ and $A_j$ are the same, then $u$ and $v$ are no longer separated in $\cC'$.  Let us use this to see what happens to the objective. Let $\chi_{uv}$ be an indicator for $u,v$ being in the same cluster in the optimal clustering $\cC$, and let $\chi_{uv}'$ be a similar indicator in $\cC'$. The original objective is
\[ C^- + \sum_{ij} \eta_{ij} \chi_{ij}. \] 

From the above reasoning, if $\chi_{uv} = 1$, then $\chi_{uv}' = 1$.  Also, if $\chi_{uv} = 0$, $\E[ \chi_{uv}' ] = 1/t$,  i.e., there is a probability precisely $1/t = \eps$ that the clusters containing $u, v$ now get the same color.  Thus, we can write the expected value of the new objective as
\[ C^- + \sum_{ij} \eta_{ij} \chi_{ij} + \eps \sum_{ij} \eta_{ij} (1-\chi_{ij}). \]

Let us  denote $S = \sum_{ij} \eta_{ij} \chi_{ij}$ and $S' = \sum_{ij} \eta_{ij} (1-\chi_{ij})$.  Then by definition, we have $S+S' = \sum_{ij} \eta_{ij} = C^+ - C^-$.  Now, to show that we have a $(1-\eps)$ approximation, we need to show that $C^- + S + \eps S' \ge (1-\eps) (C^- + S)$. This simplifies to $C^- + S + S' \ge 0$, or equivalently, $C^+ \ge 0$, which is clearly true.

This completes the proof.
\end{proof}

Once we fix the number of clusters $k$, we can write correlation clustering as a quadratic program in a natural way: for each vertex $i$, have $k$ variables $x_{i\ell}$, which is supposed to indicate if $i$ is given the label $\ell$.  We then have the constraint that $\sum_{\ell} x_{i\ell} = 1$ for all $i$.  The objective function then has a clean form:
\begin{align*}
\sum_{ij}[\sum_{\ell=1}^{k} x_{i \ell} (1 - x_{j\ell}) c^{-}_{ij} + x_{i \ell} x_{j \ell} c^{+}_{ij}] = \sum_{ij} \sum_{\ell} x_{i \ell} c^{-}_{ij} + x_{i \ell} x_{j \ell} \eta_{ij} 
=  \sum_{i, \ell} x_{i \ell} (\rho_{i \ell} + d^{-}_i),
\end{align*}
where $x_{i \ell} =1$ iff vertex $i \in C_{\ell}$, and $\rho_{i \ell} = \sum_{j \in \Gamma(i)} x_{j \ell} \eta_{ij}$ and $d^{-}_i = \sum_j c^{-}_{ij}$.

Note the similarity with the program for \maxcut. We will show that the framework from Section~\ref{sec:paper-outline} carries over with minor changes.  The details of the new $\akkest$ procedure can be found in Section~\ref{akk-est-corr} (it requires one key change: we now need to consider $k$-partitions of the seed set in order to find $\rho$).  The duality based proof is slightly more involved; however we can use the same rough outline.  The proof is presented in Section~\ref{inducedLP-CC}.  

\subsection{LP Estimation Procedure for Correlation Clustering}\label{akk-est-corr}
Let $H = (V, E, c)$ be a weighted, undirected graph with edge weights $c^{+}_{ij}$ and $c^{-}_{ij}$ as before, and let $\gamma: V \rightarrow [0,1]$ denote sampling probabilities.  We define
$\akkest_{C} (H, \gamma)$ to be the output of the following randomized algorithm: sample a set $S$ by including each vertex $i$ in it w.p. $\gamma_i$ independently; next, for each partition $(A_1, \cdots, A_k)$ of $S$, solve the LP defined below, and output the largest objective value. 

$LP_{A_1, \cdots, A_k}(V)$ is the following linear program. (As before, we use constants $\rho_{i \ell} := \sum_{j \in \Gamma(i) \cap A_{\ell}} \frac{\eta_{ij}}{\gamma_j}$.)
\begin{align*}\label{eq:akk-CC}
\text{maximize} &\quad \sum_i x_{i \ell} ( \rho_{i \ell} + d^{-}_i) - (s_{i \ell} + t_{i \ell}) \\
\text{subject to} &\quad \rho_{i \ell} - t_{i \ell} \le \sum_{j \in \Gamma(i)} \eta_{ij} x_{j \ell} \le
  \rho_{i \ell} + s_{i \ell} \quad \forall i, \ell\\ 
&\quad  \sum_{\ell} x_{i \ell} = 1 \quad \forall i \in [n]\\
&\quad s_{i \ell}, t_{i \ell} \geq 0 \quad \forall i, \ell
\end{align*}

Once again, the best choice of $s_{i \ell}, t_{i \ell}$ for each pair $i, \ell$ are so that $s_{i \ell}+ t_{i \ell}= |\rho_{i \ell} - \sum_{j \in \Gamma(i)} \eta_{ij} x_{j \ell}|$. 

We now show a result analogous to the one earlier -- that under appropriate conditions, $\akkest_C$ is approximately equal to the optimal correlation clustering objective value.

\begin{theorem} \label{estimation-cc}
Let $H$ be a weighted graph on $n$ vertices with edge weights $c_{ij}^+, c_{ij}^-$ that add up to $W$.  Suppose the sampling probabilities $\gamma_i$ satisfy the condition
\begin{equation}
\label{eq:w-condition-CC}w_{ij} \le \frac{W \eps^2}{8 k^2 \log n} \frac{\gamma_i \gamma_j}{\sum_u \gamma_u} \quad \text{for all $i,j$.}
\end{equation}
Then, we have $\akkest_C (H, \gamma) \in CC (H) \pm \eps W$, with probability at least $1 - 1/n^2$ (where the probability is over the random choice of $S$).
\end{theorem}

Note that the only difference is the $k^2$ term in~\eqref{eq:w-condition-CC}.
\begin{proof}
As before, the proof follows from two complementary claims.

\noindent {\bf Claim 1.} W.h.p. over the choice of $S$, there exists a partitioning $(A_1, \cdots, A_k)$ of $S$ such that $LP_{A_1, \cdots, A_k}(H) \ge CC(H) - \eps W$. 

\noindent {\bf Claim 2.} Consider any feasible solution to the LP above (for some values $\rho_{i \ell}, s_{i \ell}, t_{i \ell}$). There exists a partitioning in $H$ of objective value at least the LP objective.

The proof of Claim 1 mimics the proof in the case of MAXCUT. We use Bernstein's inequality for every $\ell \in [k]$ with deviation being bounded by $\eps (d_i + \Delta)/k$ in each term.  This is why we need an extra $k^2$ term in the denominator of~\eqref{eq:w-condition-CC}. We omit the details.

\begin{proof}[Proof of Claim 2.]
Suppose we have a feasible solution $x$ to the LP of objective value $\sum_{i, \ell} x_{i \ell} (d^{-}_i + \rho_{i \ell}) - \sum_{i, \ell} |\rho_{i \ell} - \sum_{j \in \Gamma(i)} \eta_{ij} x_{j \ell}|$, and we wish to move to a partitioning of at least this value. To this end, define the quadratic form
\[ Q(x) := \sum_{i, \ell} x_{i \ell} \big(d^{-}_i + \sum_{j \in \Gamma(i)} \eta_{ij} x_{j \ell} \big). \]

The first observation is that for any $x \in [0,1]^{n k}$, and any real numbers $\rho_i$, we have
\[ Q(x) \ge \sum_{i, \ell} x_{i \ell} (d^{-}_i + \rho_{i \ell}) - \sum_{i, \ell} |\rho_{i \ell} - \sum_{j \in \Gamma(i)} \eta_{ij} x_{j \ell}|.\]
This is true simply because $Q(x) = \sum_{i, \ell} x_{i \ell} \big( d^{-}_i + \rho_{i \ell} \big) - x_{i \ell} \big( \rho_{i \ell} - \sum_{j \in \Gamma(i)} \eta_{ij} x_{j \ell} \big)$, and the fact that the second term is at least $- |\rho_{i \ell} - \sum_{j \in \Gamma(i)} \eta_{ij} x_{j \ell}|$, as $x_i \in [0,1]$. 

Next, note that the maximum of the form $Q(x)$ over $[0,1]^{n k}$ has to occur at a boundary point, since for any fixing of variables other than the $i$th group of variables $x_{i 1}, \cdots, x_{i k}$ for a given $i$, the form reduces to a linear function of $x_{i \ell}$, $1 \leq \ell \leq k$, which attains maximum at one of the boundaries when subject to the constraint $\sum_{\ell} x_{i \ell} = 1$. Using this observation repeatedly for $i \in [n]$ lets us conclude that there is a $y \in \{0,1\}^{nk}$ such that $Q(y) \ge Q(x)$.  Since any such $y$ corresponds to a partitioning, and $Q(y)$ corresponds to its objective value, the claim follows. 
\end{proof}

This completes the proof of Theorem~\ref{estimation-cc}.
\end{proof}

As in the case of MAXCUT, we can show that the estimation procedure can be used for estimating both in the original graph (with uniform probabilities $q_i$), and with the graph $H$, with sampling probabilities $q_i/p_i$. 
We thus skip stating these claims formally.

\subsection{Induced Linear Programs for Correlation Clustering}\label{inducedLP-CC}
We next need to prove that the $\akkest_C$ procedures have approximately the same values on $G$ and $H$  (with appropriate $\gamma$'s).  To show this, we consider a sample $(S, S')$ drawn as before, and show that
\begin{equation}\label{eq:to-show-cc2}
\max_{(A_1, \cdots, A_k):S} LP_{A_1, \cdots, A_k}^{\gamma}(V) \ge \Delta^2 \max_{(A_1, \cdots, A_k): S} LP_{A_1, \cdots, A_k}^{\alpha}(S') - \eps n \Delta,
\end{equation}
where $\gamma_i = q_i$ and $\alpha_i = q_i/p_i$. As before, we consider the duals of the two programs. This is the main place in which our correlation clustering analysis differs from the one for MAXCUT.   
The dual is as follows
\begin{align*}
\text{minimize} \quad  \sum_i u_i + \sum_{i, \ell} \rho_{i \ell} z_{i \ell}\\
\text{subject to} \quad u_i + \sum_{j \in \Gamma(i)} \eta_{j \ell} z_{j \ell}  \geq d^{-}_i + \rho_{i \ell} \quad \forall i , \ell\\ 
\quad \quad -1 \le z_{i \ell} \leq 1 \quad \forall i \in [n], \ell \in [k]
\end{align*}

The difference now is that for any vector $z$, the optimal choice of $u_i$ is $\max_{\ell} \{d^{-}_i + \rho_{i \ell} - \sum_{j \in \Gamma(i)} \eta_{j \ell} z_{j \ell} \}$.  This is now a maximum of $k$ terms, as opposed to the max of $0$ and one other term in the case of MAXCUT.  But once again, we can think of the dual solution as being the vector $z$, and we again have $u_i \le 2d_i$.  The dual of $LP_{A_1, \cdots, A_k}^{\alpha}(H)$ can be written down similarly. 

As we did earlier, we take a solution $z$ to the dual of the LP on $G$, and use the same values (for the vertices in $S'$) as the solution to the dual on $H$.  The objective values are now as follows.
\begin{align}
\dual_G &= \sum_{i, \ell} \rho_{i \ell} z_{i \ell} + \sum_i \max_{\ell} \{ d^{-}_i + \rho_{i \ell} - \sum_{j \in \Gamma(i)} \eta_{ij} z_{j \ell} \}  \label{eq:dual-g-CC} \\
\dual_H &\le \sum_{i \in S', \ell} \rhotil_{i \ell} z_{i \ell} + \sum_i \max_{\ell} \{ ~ \dtil_{i}^{-} + \rhotil_{i \ell} - \sum_{j \in \Gamma(i) \cap S'} w_{ij} \eta_{ij} z_{j \ell} \}  \label{eq:dual-h-CC}
\end{align}
Here also, we have $\rhotil_{i\ell} = \frac{\rho_{i\ell}}{p_i \Delta^2}$.  We now show that w.p. at least $1-\frac{1}{n^4}$,
\begin{equation}\label{eq:to-show-cc}
\max_{(A_1, \dots, A_k):S} \dual_H \le \frac{1}{\Delta^2} \max_{(A_1, \dots, A_k):S} \dual_G + \frac{\eps n}{\Delta}. \end{equation}

We can use the same trick as in the MAXCUT case, and move to $\dual_H\gs$, in which we use $\rhotil_{i\ell}\gs := \frac{\rho_{i\ell}}{p_i \Delta^2}$, weights $w_{ij}\gs$ as before.  The proof of Lemma~\ref{lem:dual-error} applies verbaitm. Thus it suffices to show that  w.h.p. (assuming $S$ satisfies conditions analogous to Lemma~\ref{lem:helper-maxcut}),
\[ \max_{(A_1, \dots, A_k):S} \dual_H\gs \le \frac{1}{\Delta^2} \max_{(A_1, \dots, A_k):S} \dual_G + \frac{\eps n}{2\Delta}. \]

Consider the expression
\begin{equation}
\label{eq:dual-diff-CC}
\dual^*_H - \frac{1}{\Delta^2} \dual_G = \sum_{i} \left( \sum_{\ell} (Y_i \rhotil^*_{i \ell} z_{i \ell} - \frac{1}{\Delta^2} \cdot \rho_{i \ell} z_{i \ell}) \right) + \left( Y_i \tilde{u}^*_i - \frac{1}{\Delta^2} \cdot u_i \right).
\end{equation} 

We view this as two summations (shown by the parentheses), and bound them separately. The first is relatively easy. We observe that by definition, 
\begin{equation}\label{eq:rhotil-CC} \rhotil^*_{i \ell} = \sum_{j \in \Gamma(i) \cap C_\ell} \frac{w^*_{ij} \eta_{ij} p^*_j}{q_j} = \frac{\eta_{ij}}{p^*_i \Delta^2} \sum_{j \in\Gamma(i) \cap C_{\ell}} \frac{1}{q_i} = \frac{1}{p^*_i \Delta^2} \cdot \rho_{i \ell}.
\end{equation}

This implies that we can write the first term as $\sum_{i \ell} \frac{\rho_{i \ell} z_{i \ell}}{p^*_{i} \Delta^2} \big( Y_i  - p^*_i \big)$.
This will then be bounded via Bernstein's inequality. 

For bounding the second term, we start with the trick of splitting it into two terms by adding a ``hybrid'' term, as follows:
\[ \sum_i Y_i \tilde{u}^*_i - \frac{1}{\Delta^2} \cdot u_i = \sum_i \left( Y_i \tilde{u}^*_i - Y_i \frac{u_i}{p^*_i \Delta^2}\right) + \sum_i \left( Y_i \frac{u_i}{p^*_i \Delta^2} - \frac{1}{\Delta^2} \cdot u_i\right). \]

The second term can again be bounded using Bernstein's inequality, using the fact that $u_i = O(d_i)$.
Let us thus consider the first term. We can appeal to the fact $|\max \{ P_1, \dots, P_k \} - \max \{Q_1, \dots, Q_k\}| \le \sum_i |P_i - Q_i|$, to bound it by
\[ \sum_i Y_i \sum_{\ell} \left| ~ \dtil_{i}^{* -} + \rhotil^*_{i \ell} - \sum_{j \in \Gamma(i) \cap S'} w^*_{ij} \eta_{ij} z_{j \ell} - \frac{1}{p^*_i \Delta^2} \big( d^{-}_i + \rho_{i \ell} - \sum_{j \in \Gamma(i)} \eta_{ij} z_{j \ell}  \big) \right|. \]

Using \eqref{eq:rhotil-CC} and $w^*_{ij} = \frac{1}{p^*_i p^*_j \Delta^2}$ we can bound the above by
\begin{align*}
 &\sum_i Y_i \sum_{\ell} \left| \sum_{j \in \Gamma(i): \eta_{ij} <0} Y_j |\eta_{ij}| w^*_{ij} (1- z_{j \ell}) - \frac{|\eta_{ij}|}{p^*_i \Delta^2} c_{j \ell} \right|  + \sum_i Y_i \sum_{\ell} \left| \sum_{j \in \Gamma(i): \eta_{ij} > 0} |\eta_{ij}| w^*_{ij} z_{j \ell} (Y_j - p^*_j) \right| \qquad  \\  
& =
\sum_i Y_i  \sum_{\ell} \left| \sum_{j \in \Gamma(i): \eta_{ij} <0} |\eta_{ij}| w^*_{ij} (1 - z_{j\ell})(Y_j - p^*_j) \right|  +  \sum_i Y_i \sum_{\ell} \left| \sum_{j \in \Gamma(i): \eta_{ij} > 0} |\eta_{ij}| w^*_{ij} z_{j \ell} (Y_j - p^*_j) \right|  
\end{align*}

This leads to a sum over $\ell$ of terms of the form
\begin{equation}\label{eq:quadratic-CC}
 \sum_i \left| \sum_{j \in \Gamma(i):\eta_{ij} \ne 0} |\eta_{ij}| w_{ij}\gs  (1-z_{j\ell}) (Y_j - p_j\gs) \right|.
\end{equation}

The nice thing now is that we can appeal (in a black-box manner, using the boundedness of $\eta$ and $z$) to the concentration bound for quadratic functions in Section~\ref{sec:blm}, with $\eps$ replaced by $\eps/k$, to conclude the desired concentration inequality. For this goal, we again
use a similar conditioning as for \maxcut, which we state it here.

\npara{``Good'' conditioning.} We say that a choice of $Y$'s is {\em good} if for all $i \in V$, we have 
\[  \sum_{j \in \Gamma(i)} w^*_{ij} |\eta_{ij}| Y_j \leq \frac{\eps \Delta + 2d_i}{p^*_i \Delta^2} \]
\begin{lemma}\label{lem:allgoodCC}
Let $H$ be the weighted graph obtained after sampling with probabilities $p^*_i$. For any vertex $i \in V$, we have 
\[ \Pr \left[ \sum_{j \in \Gamma(i)} w^*_{ij} |\eta_{ij}| Y_j > \frac{\eps \Delta + 2d_i}{p^*_i \Delta^2} \right] < \frac{1}{n^4}. \]
\end{lemma}
\begin{proof}
Fix some $i \in V$, and consider $\sum_{j \in \Gamma(i)} w^*_{ij}  |\eta_{ij}| Y_j = \frac{1}{p^*_i \Delta^2} \left( \sum_{j \in \Gamma(i)} \frac{Y_j  |\eta_{ij}|}{p^*_j} \right)$. The term in the parenthesis has expectation precisely $d_i$. Thus, applying Bernstein using $\max_j \frac{1}{p^*_j} \le \frac{\alpha_\eps \Delta}{\eps}$, together with $\sum_{j \in \Gamma(i)} \frac{|\eta_{ij}| p^*_j (1-p^*_j)}{p^{*2}_j} \le d_i \max_j \frac{1}{p^*_j}$, we have
\[  \Pr \big[ \sum_{j \in \Gamma(i) \cap V_H} \frac{ |\eta_{ij}| Y_j}{p^*_j} > d_i + t \big] \le \exp \left( - \frac{\eps t^2}{(d_i + t) \alpha_\eps \Delta} \right). \]
By setting $t = (d_i + \frac{\eps \Delta}{k})$, the RHS above can be bounded by
\[ \exp \left( - \frac{\eps (d_i + \frac{\eps \Delta}{k})^2}{(2d_i + \frac{\eps \Delta}{k}) \alpha_\eps \Delta} \right) \le \exp\left( - \frac{\eps^2}{2 \alpha_\eps} \right) < \frac{1}{n^4}.\]

This completes the proof, using our choice of $\alpha_\eps$.
\end{proof}

Conditioning on the $Y$ being good, we can obtain the concentration bound for the quadratic function \eqref{eq:quadratic-CC}.

This lets us take a union bound over all possible partitions of $S$ (of which there are at most $k^n$), to obtain~\eqref{eq:to-show-cc}, which then completes the proof of the main result (Theorem~\ref{sampling}), for correlation clustering.

\section*{Acknowledgement}
We thank Michael Kapralov for helpful discussions as we embarked upon this project.
\bibliography{thesis}

\begin{thebibliography}{10}

\bibitem{agarwal2005geometric}
Pankaj~K Agarwal, Sariel Har-Peled, and Kasturi~R Varadarajan.
\newblock Geometric approximation via coresets.
\newblock {\em Combinatorial and computational geometry}, 52:1--30, 2005.

\bibitem{ahn2009graph}
Kook~Jin Ahn and Sudipto Guha.
\newblock Graph sparsification in the semi-streaming model.
\newblock In {\em International Colloquium on Automata, Languages, and
  Programming}, pages 328--338. Springer, 2009.

\bibitem{ahn2012graph}
Kook~Jin Ahn, Sudipto Guha, and Andrew McGregor.
\newblock Graph sketches: sparsification, spanners, and subgraphs.
\newblock In {\em Proceedings of the 31st symposium on Principles of Database
  Systems}, pages 5--14. ACM, 2012.

\bibitem{ahn2015correlation}
KookJin Ahn, Graham Cormode, Sudipto Guha, Andrew McGregor, and Anthony Wirth.
\newblock Correlation clustering in data streams.
\newblock In {\em International Conference on Machine Learning}, pages
  2237--2246, 2015.

\bibitem{ailon2012note}
Nir Ailon and Zohar Karnin.
\newblock A note on: No need to choose: How to get both a ptas and sublinear
  query complexity.
\newblock {\em arXiv preprint arXiv:1204.6588}, 2012.

\bibitem{alon2003random}
Noga Alon, W~Fernandez De~La~Vega, Ravi Kannan, and Marek Karpinski.
\newblock Random sampling and approximation of max-csps.
\newblock {\em Journal of computer and system sciences}, 67(2):212--243, 2003.

\bibitem{alon2009combinatorial}
Noga Alon, Eldar Fischer, Ilan Newman, and Asaf Shapira.
\newblock A combinatorial characterization of the testable graph properties:
  it's all about regularity.
\newblock {\em SIAM Journal on Computing}, 39(1):143--167, 2009.

\bibitem{andoni2011streaming}
Alexandr Andoni, Robert Krauthgamer, and Krzysztof Onak.
\newblock Streaming algorithms via precision sampling.
\newblock In {\em Foundations of Computer Science (FOCS), 2011 IEEE 52nd Annual
  Symposium on}, pages 363--372. IEEE, 2011.

\bibitem{andoni2014sketching}
Alexandr Andoni, Robert Krauthgamer, and David~P Woodruff.
\newblock The sketching complexity of graph cuts.
\newblock {\em arXiv preprint arXiv:1403.7058}, 2014.

\bibitem{arora1995polynomial}
Sanjeev Arora, David Karger, and Marek Karpinski.
\newblock Polynomial time approximation schemes for dense instances of np-hard
  problems.
\newblock In {\em Proceedings of the twenty-seventh annual ACM symposium on
  Theory of computing}, pages 284--293. ACM, 1995.

\bibitem{bansal2004correlation}
Nikhil Bansal, Avrim Blum, and Shuchi Chawla.
\newblock Correlation clustering.
\newblock {\em Machine Learning}, 56(1-3):89--113, 2004.

\bibitem{barak2011}
Boaz Barak, Moritz Hardt, Thomas Holenstein, and David Steurer.
\newblock Subsampling mathematical relaxations and average-case complexity.
\newblock In {\em Proceedings of the Twenty-second Annual ACM-SIAM Symposium on
  Discrete Algorithms}, SODA '11, pages 512--531, Philadelphia, PA, USA, 2011.
  Society for Industrial and Applied Mathematics.
\newblock URL: \url{http://dl.acm.org/citation.cfm?id=2133036.2133077}.

\bibitem{boucheron2003concentration}
St{\'e}phane Boucheron, G{\'a}bor Lugosi, and Pascal Massart.
\newblock Concentration inequalities using the entropy method.
\newblock {\em Annals of Probability}, pages 1583--1614, 2003.

\bibitem{cormode2005improved}
Graham Cormode and Shan Muthukrishnan.
\newblock An improved data stream summary: the count-min sketch and its
  applications.
\newblock {\em Journal of Algorithms}, 55(1):58--75, 2005.

\bibitem{dubhashi2009concentration}
Devdatt~P Dubhashi and Alessandro Panconesi.
\newblock {\em Concentration of measure for the analysis of randomized
  algorithms}.
\newblock Cambridge University Press, 2009.

\bibitem{feige2002optimality}
Uriel Feige and Gideon Schechtman.
\newblock On the optimality of the random hyperplane rounding technique for max
  cut.
\newblock {\em Random Structures \& Algorithms}, 20(3):403--440, 2002.

\bibitem{fotakis2015sub}
Dimitris Fotakis, Michael Lampis, and Vangelis~Th Paschos.
\newblock Sub-exponential approximation schemes for csps: from dense to almost
  sparse.
\newblock {\em arXiv preprint arXiv:1507.04391}, 2015.

\bibitem{frieze1996regularity}
Alan Frieze and Ravi Kannan.
\newblock The regularity lemma and approximation schemes for dense problems.
\newblock In {\em Foundations of Computer Science, 1996. Proceedings., 37th
  Annual Symposium on}, pages 12--20. IEEE, 1996.

\bibitem{giotis2006correlation}
Ioannis Giotis and Venkatesan Guruswami.
\newblock Correlation clustering with a fixed number of clusters.
\newblock In {\em Proceedings of the seventeenth annual ACM-SIAM symposium on
  Discrete algorithm}, pages 1167--1176. Society for Industrial and Applied
  Mathematics, 2006.

\bibitem{goel2010graph}
Ashish Goel, Michael Kapralov, and Sanjeev Khanna.
\newblock Graph sparsification via refinement sampling.
\newblock {\em arXiv preprint arXiv:1004.4915}, 2010.

\bibitem{goel2012single}
Ashish Goel, Michael Kapralov, and Ian Post.
\newblock Single pass sparsification in the streaming model with edge
  deletions.
\newblock {\em arXiv preprint arXiv:1203.4900}, 2012.

\bibitem{goldreich1998property}
Oded Goldreich, Shari Goldwasser, and Dana Ron.
\newblock Property testing and its connection to learning and approximation.
\newblock {\em Journal of the ACM (JACM)}, 45(4):653--750, 1998.

\bibitem{jowhari2011tight}
Hossein Jowhari, Mert Sa{\u{g}}lam, and G{\'a}bor Tardos.
\newblock Tight bounds for lp samplers, finding duplicates in streams, and
  related problems.
\newblock In {\em Proceedings of the thirtieth ACM SIGMOD-SIGACT-SIGART
  symposium on Principles of database systems}, pages 49--58. ACM, 2011.

\bibitem{kapralov2013better}
Michael Kapralov.
\newblock Better bounds for matchings in the streaming model.
\newblock In {\em SODA}, pages 1679--1697. SIAM, 2013.

\bibitem{kapralov2014approximating}
Michael Kapralov, Sanjeev Khanna, and Madhu Sudan.
\newblock Approximating matching size from random streams.
\newblock In {\em SODA}, pages 734--751. SIAM, 2014.

\bibitem{kapralov2015streaming}
Michael Kapralov, Sanjeev Khanna, and Madhu Sudan.
\newblock Streaming lower bounds for approximating max-cut.
\newblock In {\em SODA}, pages 1263--1282. SIAM, 2015.

\bibitem{kapralov20171}
Michael Kapralov, Sanjeev Khanna, Madhu Sudan, and Ameya Velingker.
\newblock (1+ $\omega$ (1))-approximation to max-cut requires linear space.
\newblock In {\em Proceedings of the Twenty-Eighth Annual ACM-SIAM Symposium on
  Discrete Algorithms}, pages 1703--1722. SIAM, 2017.

\bibitem{kapralov2014single}
Michael Kapralov, Yin~Tat Lee, Cameron Musco, Christopher Musco, and Aaron
  Sidford.
\newblock Single pass spectral sparsification in dynamic streams.
\newblock In {\em Foundations of Computer Science (FOCS), 2014 IEEE 55th Annual
  Symposium on}, pages 561--570. IEEE, 2014.

\bibitem{kelner2013spectral}
Jonathan~A Kelner and Alex Levin.
\newblock Spectral sparsification in the semi-streaming setting.
\newblock {\em Theory of Computing Systems}, 53(2):243--262, 2013.

\bibitem{bhm4}
Dmitry Kogan and Robert Krauthgamer.
\newblock Sketching cuts in graphs and hypergraphs.
\newblock In {\em Proc. ITCS}, pages 367--376. ACM, 2015.

\bibitem{mathieu2008yet}
Claire Mathieu and Warren Schudy.
\newblock Yet another algorithm for dense max cut: go greedy.
\newblock In {\em Proceedings of the nineteenth annual ACM-SIAM symposium on
  Discrete algorithms}, pages 176--182. Society for Industrial and Applied
  Mathematics, 2008.

\bibitem{monemizadeh20101}
Morteza Monemizadeh and David~P Woodruff.
\newblock 1-pass relative-error lp-sampling with applications.
\newblock In {\em Proceedings of the twenty-first annual ACM-SIAM symposium on
  Discrete Algorithms}, pages 1143--1160. SIAM, 2010.

\bibitem{rudelson2007}
Mark Rudelson and Roman Vershynin.
\newblock Sampling from large matrices: An approach through geometric
  functional analysis.
\newblock {\em J. ACM}, 54(4), July 2007.
\newblock URL: \url{http://doi.acm.org/10.1145/1255443.1255449}, \href
  {http://dx.doi.org/10.1145/1255443.1255449}
  {\path{doi:10.1145/1255443.1255449}}.

\end{thebibliography}







\end{document}